\documentclass[12pt,oneside]{amsart}

\usepackage{amsmath, amssymb, amsthm, amscd, 
}

\usepackage{amsaddr}

\allowdisplaybreaks[4]

\newcommand{\eval}[2][\right]{\relax
  \ifx#1\right\relax \left.\fi#2#1\rvert}

\newcommand{\pd}{{\partial}}

\newcommand{\la}{{\lambda}}

\newcommand{\ost}{\mathbb{U}}

\newcommand{\bG}{\mathbf{G}}

\newcommand{\er}{\eqref}
\newcommand{\cl}{\colon}

\newcommand{\beq}{\begin{equation}}
\newcommand{\ee}{\end{equation}}

\newcommand{\bmu}{\begin{multline}}
\newcommand{\emul}{\end{multline}}

\newcommand{\frl}{\mathfrak{F}}
\newcommand{\frid}{\mathfrak{I}}
\newcommand{\fla}{\mathbf{A}}
\newcommand{\flb}{\mathbf{B}}
\newcommand{\flz}{\mathbf{Z}}

\newcommand{\wga}{\mathcal{A}}
\newcommand{\wgb}{\mathcal{B}}
\newcommand{\wea}{\mathfrak{W}_a}
\newcommand{\zf}{\mathbf{F}}
\newcommand{\zs}{\mathbf{S}}
\newcommand{\zd}{\mathrm{D}}
\newcommand{\xtf}{\mathbb{L}}
\newcommand{\swe}{\mathfrak{R}}
\newcommand{\ds}{\mathbf{S}}

\newcommand{\anA}{\mathsf{A}}
\newcommand{\anB}{\mathsf{B}}

\newcommand{\CE}{\mathcal{E}}
\newcommand{\ce}{\mathcal{E}}

\newcommand{\zp}{\mathbb{Z}_{\ge 0}}
\newcommand{\zsp}{\mathbb{Z}_{>0}}

\newcommand{\ad}{{\rm ad\,}}

\newcommand{\msl}{\mathfrak{sl}}
\newcommand{\gl}{\mathfrak{gl}}

\newcommand{\mg}{\mathfrak{g}}

\newcommand{\agn}{\mathfrak{A}}

\newcommand{\bl}{\mathfrak{L}}

\newcommand{\ga}{\mathbb{A}}
\newcommand{\gb}{\mathbb{B}}

\newcommand{\lb}{\label}

\newcommand{\vf}{\varphi}

\newcommand{\Com}{\mathbb{C}}
\newcommand{\fik}{\mathbb{K}}

\newcommand{\un}{\mathrm{U}}

\DeclareMathOperator{\fd}{\mathbb{F}}
\DeclareMathOperator{\fds}{\mathbb{F}}

\newcommand{\zcs}{\mathcal{V}}

\newcommand{\oc}{p}
\newcommand{\ocs}{p}
\newcommand{\sm}{N}
\newcommand{\hrf}{\mu}

\newcommand{\eo}{d}
\newcommand{\kd}{q}
\newcommand{\ic}{m}

\newtheorem{theorem}{Theorem}

\newtheorem{lemma}{Lemma}

\theoremstyle{definition}

\newtheorem{example}{Example}
\newtheorem{remark}{Remark}

\begin{document}

\title[Lie algebras responsible for zero-curvature representations]
{Lie algebras responsible for zero-curvature representations of 
scalar evolution equations}
\date{}

\author{Sergei Igonin}
\address{Center of Integrable Systems, 
P.G. Demidov Yaroslavl State University, Yaroslavl, Russia, \\
INdAM, Dipartimento di Scienze Matematiche, Politecnico di Torino, \\
Corso Duca degli Abruzzi 24, 10129 Torino, Italy\\
\textup{E-mail address: s-igonin@yandex.ru}}

\author{Gianni Manno}
\address{Dipartimento di Scienze Matematiche, Politecnico di Torino, \\
Corso Duca degli Abruzzi 24, 10129 Torino, Italy\\
\textup{E-mail address: giovanni.manno@polito.it}}

\begin{abstract}
Zero-curvature representations (ZCRs) are one 
of the main tools in the theory of integrable PDEs. 
In particular, Lax pairs for (1+1)-dimensional PDEs can be interpreted as ZCRs. 

For any (1+1)-dimensional scalar evolution equation $\CE$, 
we define a family of Lie algebras $\fds(\CE)$ 
which are responsible for all ZCRs of $\CE$ in the following sense.
Representations of the algebras $\fds(\CE)$ 
classify all ZCRs of the equation $\CE$ up to local gauge transformations. 
To achieve this, we find a normal form for ZCRs with respect 
to the action of the group of local gauge transformations.

As we show in other publications, using these algebras, 
one obtains some necessary conditions for integrability of the considered PDEs    
(where integrability is understood in the sense of soliton theory)
and necessary conditions for existence of a B\"acklund transformation   
between two given equations.
Examples of proving non-integrability and 
applications to obtaining non-existence results for B\"acklund transformations
are presented in other publications as well.



In our approach, ZCRs may depend on partial derivatives of arbitrary order, 
which may be higher than the order of the equation $\CE$.
The algebras $\fds(\CE)$ 
generalize Wahlquist-Estabrook prolongation algebras, 
which are responsible for a much smaller class of ZCRs.

In this paper we describe general properties of $\fds(\CE)$  
and present generators and relations for these algebras.
In other publications we study the structure of $\fds(\CE)$   
for equations of
KdV, Krichever-Novikov, Kaup-Kupershmidt, Sawada-Kotera types.
Among the obtained algebras, one finds  
infinite-dimensional Lie algebras of certain matrix-valued functions 
on rational and elliptic algebraic curves.
\end{abstract}

\keywords{Scalar evolution equations, 
zero-curvature representations,
gauge transformations,
normal forms for zero-curvature representations,
infinite-dimensional Lie algebras} 

\subjclass[2010]{37K30, 37K35}

\maketitle

\section{Introduction}
\lb{subsint1}

Zero-curvature representations and B\"acklund transformations 
belong to the main tools in the theory of integrable PDEs 
(see, e.g.,~\cite{ft,backl,zakh-shab}).
This paper along with~\cite{zcrm17,sbt18,sint18} is part of a research program
on investigating the structure of zero-curvature representations (ZCRs)
for partial differential equations (PDEs) of various types. 
The study of ZCRs performed in this paper leads to some results on 
B\"acklund transformations and integrability, which are described in~\cite{sbt18,sint18}.

Here we study (1+1)-dimensional scalar evolution equations
\beq
\label{eveq_intr}
u_t=F(x,t,u_0,u_1,\dots,u_{\eo}),\qquad\quad  
u=u(x,t),
\ee
where we use the notation
\beq
\lb{dnot}
u_t=\frac{\pd u}{\pd t},\qquad\quad u_0=u,\qquad\quad
u_k=\frac{\pd^k u}{\pd x^k},\qquad\quad
k\in\zp.
\ee
The number $\eo\ge 1$ in~\er{eveq_intr} is such that 
the function $F$ may depend only on $x$, $t$, $u_k$ for $k\le\eo$. 
The symbol $\zp$ denotes the set of nonnegative integers.

Methods of this paper can also be applied
to (1+1)-dimensional multicomponent evolution PDEs, see~\cite{zcrm17}.

\begin{remark}
\lb{rfxtu}
When we consider a function $Q=Q(x,t,u_0,u_1,\dots,u_l)$ 
for some $l\in\zp$, we always assume that this function is analytic 
on an open subset of the space with the coordinates 
$x,t,u_0,u_1,\dots,u_l$.
For example, $Q$ may be a meromorphic function,  
because a meromorphic function is analytic on some open subset.
\end{remark}

PDEs of the form~\er{eveq_intr} have attracted a lot of attention in the last $50$ years and 
have been a source of many remarkable results on integrability. 
In particular, some types of equations~\er{eveq_intr} possessing 
higher-order symmetries and conservation laws have been 
classified (see, e.g.,~\cite{mesh-sok13,mikh91,sand-wang2009} and references therein).
However, the problem of complete understanding of all integrability properties 
for equations~\er{eveq_intr} is still far from being solved.  

Examples of integrable PDEs of the form~\er{eveq_intr}  
include the Korteweg-de Vries (KdV), Krichever-Novikov~\cite{krich80,svin-sok83}, 
Kaup-Kupershmidt~\cite{kaup80}, Sawada-Kotera~\cite{sk74} (Caudrey-Dodd-Gibbon~\cite{cdg76}) 
equations. 
Many more examples can be found in~\cite{mesh-sok13,mikh91,sand-wang2009} 
and references therein. 

In the present paper, integrability is understood 
in the sense of soliton theory and the inverse scattering method. 
(This is sometimes called S-integrability.) 
It is well known that, in order to investigate possible 
integrability properties of~\er{eveq_intr}, 
one needs to consider ZCRs.
(In particular, Lax pairs for equations~\er{eveq_intr} 
can be interpreted as ZCRs.) 

Let $\mg$ be a finite-dimensional Lie algebra.  
For an equation of the form~\er{eveq_intr}, 
a \emph{zero-curvature representation \textup{(}ZCR\textup{)} 
with values in~$\mg$} is given by $\mg$-valued functions
\beq
\lb{mnoc}
A=A(x,t,u_0,u_1,\dots,u_\ocs),\qquad\quad B=B(x,t,u_0,u_1,\dots,u_{\ocs+\eo-1})
\ee
satisfying
\beq
\lb{mnzcr}
D_x(B)-D_t(A)+[A,B]=0.
\ee

The \emph{total derivative operators} $D_x$, $D_t$ in~\er{mnzcr} are 
\beq
\lb{evdxdt}
D_x=\frac{\pd}{\pd x}+\sum_{k\ge 0} u_{k+1}\frac{\pd}{\pd u_k},\qquad\qquad
D_t=\frac{\pd}{\pd t}+\sum_{k\ge 0} D_x^k\big(F(x,t,u_0,u_1,\dots,u_{\eo})\big)\frac{\pd}{\pd u_k}.
\ee

The number $\ocs$ in~\er{mnoc} is such that  
the function $A$ may depend only on the variables $x$, $t$, $u_{k}$ for $k\le\ocs$.
Then equation~\er{mnzcr} implies that 
the function $B$ may depend only on $x$, $t$, $u_{k'}$ for $k'\le\ocs+\eo-1$.

Such ZCRs are said to be \emph{of order~$\le\ocs$}. 
In other words, a ZCR given by $A$, $B$ is of order~$\le\ocs$ iff 
$\dfrac{\pd A}{\pd u_l}=0$ for all $l>\ocs$. 

The right-hand side $F=F(x,t,u_0,\dots,u_{\eo})$ 
of~\er{eveq_intr} appears in condition~\er{mnzcr}, 
because $F$ appears in the formula for the operator $D_t$ in~\er{evdxdt}.
Note that~\er{mnzcr} can be written as $[D_x+A,\,D_t+B]=0$, because $[D_x,D_t]=0$.
See also Remark~\ref{als} for another interpretation of equation~\er{mnzcr}.

We study the following problem. 
How to describe all ZCRs~\er{mnoc},~\er{mnzcr} for a given equation~\er{eveq_intr}? 

In the case when $\ocs=0$ and the functions $F$, $A$, $B$ do not depend on $x$, $t$, 
a partial answer to this question is provided by the Wahlquist-Estabrook prolongation
method (WE method for short). 
Namely, for a given equation of the form $u_t=F(u_0,u_1,\dots,u_{\eo})$, 
the WE method constructs a Lie algebra so that ZCRs of the form
\beq
\lb{wecov}
A=A(u_0),\qquad B=B(u_0,u_1,\dots,u_{\eo-1}),\qquad
D_x(B)-D_t(A)+[A,B]=0
\ee
correspond to representations of this algebra (see, e.g.,~\cite{dodd,mll-2012,nonl,Prol}). 
It is called the \emph{Wahlquist-Estabrook prolongation algebra}.
Note that in~\er{wecov} the function $A=A(u_0)$ depends only on~$u_0$.

To study the general case of ZCRs~\er{mnoc},~\er{mnzcr} with arbitrary $\ocs$ 
for any equation~\er{eveq_intr}, 
we need to consider gauge transformations. 

Without loss of generality, one can assume that $\mg$ is a Lie subalgebra 
of $\gl_\sm$ for some $\sm\in\zsp$, where $\gl_\sm$ is the algebra of 
$\sm\times\sm$ matrices with entries from $\mathbb{R}$ or $\mathbb{C}$. 
So our considerations are applicable to both cases $\gl_\sm=\gl_\sm(\mathbb{R})$ 
and $\gl_\sm=\gl_\sm(\mathbb{C})$.
And we denote by $\mathrm{GL}_\sm$ the group of invertible $\sm\times\sm$ matrices.

Let $\fik$ be either $\Com$ or $\mathbb{R}$.
Then $\gl_\sm=\gl_\sm(\fik)$ and $\mathrm{GL}_\sm=\mathrm{GL}_\sm(\fik)$.
In this paper, all algebras are supposed to be over the field~$\fik$.

\begin{remark}
\lb{als}
So we suppose that functions $A$, $B$ in~\er{mnzcr} 
take values in $\mg\subset\gl_\sm$.
Then condition~\er{mnzcr} implies that the auxiliary linear system 
$$
\pd_x(W)=-AW,\qquad\quad
\pd_t(W)=-BW
$$
is compatible modulo~\er{eveq_intr}.
Here $W=W(x,t)$ is an invertible $\sm\times\sm$ matrix-function.
\end{remark}

Let $\mathcal{G}\subset\mathrm{GL}_\sm$ be the connected matrix Lie group 
corresponding to the Lie algebra $\mg\subset\gl_\sm$.
(That is, $\mathcal{G}$ is the connected 
immersed Lie subgroup of $\mathrm{GL}_\sm$ 
corresponding to the Lie subalgebra $\mg\subset\gl_\sm$.)
A \emph{gauge transformation} is given by a matrix-function 
$G=G(x,t,u_0,u_1,\dots,u_l)$ with values in~$\mathcal{G}$.

For any ZCR~\er{mnoc},~\er{mnzcr} and 
any gauge transformation $G=G(x,t,u_0,\dots,u_l)$, the functions 
\beq
\lb{mnprint}
\tilde{A}=GAG^{-1}-D_x(G)\cdot G^{-1},\qquad\qquad
\tilde{B}=GBG^{-1}-D_t(G)\cdot G^{-1}
\ee
satisfy $D_x(\tilde{B})-D_t(\tilde{A})+[\tilde{A},\tilde{B}]=0$ and, therefore, form a ZCR.
Moreover, since $A$, $B$ take values in~$\mg$ and $G$ 
takes values in~$\mathcal{G}$, 
the functions $\tilde{A}$, $\tilde{B}$ take values in~$\mg$.
(This is well known, but for completeness we prove this in Lemma~\ref{lemgt}.)

The ZCR~\er{mnprint} is said to be \emph{gauge equivalent} to the ZCR~\er{mnoc},~\er{mnzcr}. 
For a given equation~\er{eveq_intr}, formulas~\er{mnprint} determine an action of the group 
of gauge transformations on the set of ZCRs of this equation.

\begin{remark}
\lb{rogt}
So we study gauge transformations with values in~$\mathcal{G}$.
Alternatively, one can take some other Lie group 
$\tilde{\mathcal{G}}\subset\mathrm{GL}_\sm$ whose Lie algebra is $\mg$
and consider gauge transformations with values in~$\tilde{\mathcal{G}}$.
The results of this paper will remain valid, if one replaces 
$\mathcal{G}$ by $\tilde{\mathcal{G}}$ everywhere.
\end{remark}

The WE method does not use gauge transformations in a systematic way. 
In the classification of ZCRs~\er{wecov} this is acceptable, 
because the class of ZCRs~\er{wecov} is relatively small.  

The class of ZCRs~\er{mnoc},~\er{mnzcr} is much larger than that of~\er{wecov}.
Gauge transformations play a very important role in the classification 
of ZCRs~\er{mnoc},~\er{mnzcr}. 
Because of this, the classical WE method does not produce satisfactory results 
for~\er{mnoc},~\er{mnzcr}, especially in the case~$\ocs>0$.  

To overcome this problem,
we find a normal form for ZCRs~\er{mnoc},~\er{mnzcr}
with respect to the action of the group of gauge transformations.
Using the normal form of ZCRs, for any given equation~\er{eveq_intr},
we define a Lie algebra $\fds^{\ocs}$ for each $\ocs\in\zp$ so  
that the following property holds. 

For every finite-dimensional Lie algebra $\mg$, 
any $\mg$-valued ZCR~\er{mnoc},~\er{mnzcr} of order~$\le\ocs$  
is locally gauge equivalent to the ZCR arising from a homomorphism 
$\fds^{\ocs}\to\mg$. 

More precisely, as is discussed below, 
we define a Lie algebra $\fds^{\ocs}$ for each $\ocs\in\zp$ 
and each point~$a$ of the infinite prolongation~$\CE$ of equation~\er{eveq_intr}. 
So the full notation for the algebra is $\fds^{\ocs}(\CE,a)$.  
(The family of Lie algebras $\fds(\CE)$ mentioned in the abstract 
of this paper consists of the algebras $\fds^{\ocs}(\CE,a)$ 
for all $\ocs\in\zp$, $a\in\CE$.)

Recall that the \emph{infinite prolongation} $\CE$ of equation~\er{eveq_intr} 
is an infinite-dimensional manifold with the coordinates 
$x$, $t$, $u_k$ for $k\in\zp$.
The precise definitions of the manifold $\CE$ and the algebras $\fds^{\ocs}(\CE,a)$ 
for any equation~\er{eveq_intr} are presented in Section~\ref{csev}. 
For every $\ocs\in\zp$ and $a\in\CE$, 
the algebra $\fds^{\ocs}(\CE,a)$ is defined in terms of generators and relations.
(To clarify the main idea, in Example~\ref{edfoc1} we consider the case $\ocs=1$.)

For every finite-dimensional Lie algebra $\mg$, 
homomorphisms $\fds^{\ocs}(\CE,a)\to\mg$ classify (up to gauge equivalence) 
all $\mg$-valued ZCRs~\er{mnoc},~\er{mnzcr} of order~$\le\ocs$, 
where functions $A$, $B$ are defined on a neighborhood of 
the point $a\in\CE$. See Section~\ref{csev} for details. 

According to Section~\ref{csev}, 
the algebras $\fds^{\ocs}(\CE,a)$ for $\ocs\in\zp$ 
are arranged in a sequence of surjective homomorphisms 
\beq
\lb{intfdoc1}
\dots\to\fds^{\ocs}(\CE,a)\to\fds^{\ocs-1}(\CE,a)\to\dots\to\fds^1(\CE,a)\to\fds^0(\CE,a).
\ee
According to Theorem~\ref{thzcrfd}, for each $\oc\in\zsp$, 
the algebra $\fds^{\ocs}(\CE,a)$ is responsible for ZCRs of order $\le\ocs$, 
and the algebra $\fds^{\ocs-1}(\CE,a)$ is responsible for ZCRs of order $\le\ocs-1$.    
The surjective homomorphism $\fds^{\ocs}(\CE,a)\to\fds^{\ocs-1}(\CE,a)$ in~\er{intfdoc1}
reflects the fact that any ZCR of order $\le\ocs-1$ is at the same time of order~$\le\ocs$.
The homomorphism $\fds^{\ocs}(\CE,a)\to\fds^{\ocs-1}(\CE,a)$ is 
defined by formulas~\er{fdhff}, using generators of the algebras 
$\fds^{\ocs}(\CE,a)$, $\fds^{\ocs-1}(\CE,a)$.

As we show in the preprints~\cite{sint18,sbt18},
using $\fds^{\ocs}(\CE,a)$, 
one obtains some necessary conditions for integrability of equations~\er{eveq_intr}
and necessary conditions for existence of a B\"acklund transformation 
between two given equations. 
To get such results, one needs to study certain properties 
of ZCRs~\er{mnoc},~\er{mnzcr} with arbitrary~$\ocs$, 
and we do this by means of the algebras $\fds^{\ocs}(\CE,a)$.
As explained above, the classical WE method 
(which studies ZCRs of the form~\er{wecov}) is not sufficient for this.

Applications of $\fds^{\ocs}(\CE,a)$ to obtaining 
necessary conditions for integrability of equations~\er{eveq_intr} 
are presented in~\cite{sint18}. 
Examples of the use of these conditions in 
proving non-integrability for some equations of order~$5$
are presented in~\cite{sint18} as well.
Applications to obtaining non-existence results for B\"acklund transformations
between two given equations are described in~\cite{sbt18}. 

In this paper and in~\cite{sint18,sbt18} 
we present also a number of results on the structure 
of the algebras $\fds^{\ocs}(\CE,a)$ 
for some classes of scalar evolution equations of orders~$3$,~$5$,~$7$ 
and concrete examples.
In particular, the KdV equation is considered in Theorem~\ref{thkdvm} in this paper.
The Krichever-Novikov equation is discussed in~\cite{sbt18}.
In~\cite{sint18,sbt18} we study also the algebras $\fds^{\ocs}(\CE,a)$ 
and integrability properties for a parameter-dependent $5$th-order scalar evolution equation,
which was considered by A.~P.~Fordy~\cite{fordy-hh} 
in connection with the H\'enon-Heiles system. 
The problem to study this equation was suggested to us by A.~P.~Fordy.

Relations of the algebras $\fds^{\ocs}(\CE,a)$ with parameter-dependent ZCRs
are discussed in~\cite{sint18}.

We suppose that the variables $x$, $t$, $u_k$ take values in $\fik$. 
A point $a\in\CE$ is determined by the values of the coordinates 
$x$, $t$, $u_k$ at $a$. Let
\begin{equation}
\notag
a=(x=x_a,\,t=t_a,\,u_k=a_k)\,\in\,\CE,\qquad\qquad x_a,\,t_a,\,a_k\in\fik,\qquad k\in\zp,
\end{equation}
be a point of $\CE$.
In other words, the constants $x_a$, $t_a$, $a_k$ are the coordinates 
of the point $a\in\CE$ in the coordinate system $x$, $t$, $u_k$.

\begin{example}
\lb{edfoc1}
To clarify the definition of $\fds^{\ocs}(\CE,a)$, 
let us consider the case $\ocs=1$. 
To this end, 
we fix an equation~\er{eveq_intr} and study ZCRs of order~$\le 1$ of this equation. 

According to Theorem~\ref{evcov}, any ZCR of order~$\le 1$ 
\beq
\lb{zcru1}
A=A(x,t,u_0,u_1),\qquad B=B(x,t,u_0,u_1,\dots,u_{\eo}),\qquad
D_x(B)-D_t(A)+[A,B]=0
\ee
on a neighborhood of $a\in\CE$ is gauge equivalent to a ZCR of the form 
\begin{gather}
\lb{nfzcr}
\tilde{A}=\tilde{A}(x,t,u_0,u_1),\qquad \tilde{B}=\tilde{B}(x,t,u_0,u_1,\dots,u_{\eo}),\\
\lb{nfzcreq}
D_x(\tilde{B})-D_t(\tilde{A})+[\tilde{A},\tilde{B}]=0,\\
\lb{nfab}
\frac{\pd\tilde{A}}{\pd u_1}(x,t,u_0,a_1)=0,\qquad 
\tilde{A}(x,t,a_0,a_1)=0,\qquad\tilde{B}(x_a,t,a_0,a_1,\dots,a_{\eo})=0. 
\end{gather}
Moreover, according to Theorem~\ref{tguniq},
for any given ZCR of the form~\er{zcru1}, 
on a neighborhood of $a\in\CE$ there is a unique 
gauge transformation $G=G(x,t,u_0,\dots,u_l)$ such that
the functions $\tilde{A}=GAG^{-1}-D_x(G)\cdot G^{-1}$,
$\tilde{B}=GBG^{-1}-D_t(G)\cdot G^{-1}$ 
satisfy~\er{nfzcr},~\er{nfzcreq},~\er{nfab}  
and $G(x_a,t_a,a_0,\dots,a_l)=\mathrm{Id}$, 
where $\mathrm{Id}\in\mathrm{GL}_\sm$ is the identity matrix.
Therefore, we can say that properties~\er{nfab} determine a 
normal form for ZCRs~\er{zcru1} 
with respect to the action of the group of gauge transformations 
on a neighborhood of $a\in\CE$.

A similar normal form for ZCRs~\er{mnoc},~\er{mnzcr} 
with arbitrary $\ocs$ is described in Theorem~\ref{evcov} 
and Remark~\ref{rnfzcr}.

Since the functions $\tilde{A}$, $\tilde{B}$ from~\er{nfzcr},~\er{nfab} 
are analytic on a neighborhood of $a\in\CE$, these functions
are represented as absolutely convergent power series
\begin{gather}
\label{aser1}
\tilde{A}=\sum_{l_1,l_2,i_0,i_1\ge 0} 
(x-x_a)^{l_1} (t-t_a)^{l_2}(u_0-a_0)^{i_0}(u_1-a_1)^{i_1}\cdot
\tilde{A}^{l_1,l_2}_{i_0,i_1},\\
\lb{bser1}
\tilde{B}=\sum_{l_1,l_2,j_0,\dots,j_{\eo}\ge 0} 
(x-x_a)^{l_1} (t-t_a)^{l_2}(u_0-a_0)^{j_0}\dots(u_{\eo}-a_{\eo})^{j_{\eo}}\cdot
\tilde{B}^{l_1,l_2}_{j_0\dots j_{\eo}}.
\end{gather}
Here $\tilde{A}^{l_1,l_2}_{i_0,i_1}$ and $\tilde{B}^{l_1,l_2}_{j_0\dots j_{\eo}}$ 
are elements of a Lie algebra, which we do not specify yet. 

Using formulas~\er{aser1},~\er{bser1}, we see that properties~\er{nfab} are equivalent to 
\beq
\lb{ab000int}
\tilde{A}^{l_1,l_2}_{i_0,1}=
\tilde{A}^{l_1,l_2}_{0,0}=
\tilde{B}^{0,l_2}_{0\dots 0}=0
\qquad\qquad\forall\,l_1,l_2,i_0\in\zp.
\ee
To define $\fds^1(\CE,a)$, we regard $\tilde{A}^{l_1,l_2}_{i_0,i_1}$, 
$\tilde{B}^{l_1,l_2}_{j_0\dots j_{\eo}}$ from~\er{aser1},~\er{bser1} 
as abstract symbols.  
By definition, the Lie algebra $\fds^1(\CE,a)$ is generated by the symbols 
$\tilde{A}^{l_1,l_2}_{i_0,i_1}$, $\tilde{B}^{l_1,l_2}_{j_0\dots j_{\eo}}$
for $l_1,l_2,i_0,i_1,j_0,\dots,j_{\eo}\in\zp$.
Relations for these generators are provided by equations~\er{nfzcreq},~\er{ab000int}. 
A more detailed description of this construction 
is given in Section~\ref{csev}.
\end{example}

As discussed above, the algebra~$\fds^{\ocs}(\CE,a)$ is defined 
by a certain set of generators and relations arising from 
a normal form of ZCRs. 
In Theorem~\ref{lemgenfdq} 
we describe a smaller subset of generators for~$\fds^{\ocs}(\CE,a)$. 
\begin{example}
Consider the case $\ocs=1$. 
According to the above definition of~$\fds^1(\CE,a)$, 
the algebra $\fds^1(\CE,a)$ is given by the generators  
$\tilde{A}^{l_1,l_2}_{i_0,i_1}$, $\tilde{B}^{l_1,l_2}_{j_0\dots j_{\eo}}$
and the relations arising from~\er{nfzcreq},~\er{ab000int}.
According to Theorem~\ref{lemgenfdq}, the algebra 
$\fds^1(\CE,a)$ coincides with the subalgebra generated by 
$\tilde{A}^{l_1,0}_{i_0,i_1}$ for $l_1,i_0,i_1\in\zp$, 
and a similar result is valid also for $\fds^{\ocs}(\CE,a)$ for every $\ocs$. 
\end{example}

This result helps us to describe the structure of $\fds^{\ocs}(\CE,a)$ 
and the homomorphisms~\er{intfdoc1} more explicitly for some PDEs.
Consider equations of the form 
\beq
\lb{utukd}
u_t=u_{2\kd+1}+f(x,t,u_0,u_1,\dots,u_{2\kd-1}),\qquad\qquad \kd\in\{1,2,3\},
\ee
where $f$ is an arbitrary function. 
Examples of such PDEs include 
\begin{itemize}
\item the KdV equation $u_t=u_3+u_0u_1$, 
\item the Kaup-Kupershmidt equation~\cite{kaup80} $u_t=u_5+10u_0u_3+25u_1u_2+20u_0^2u_1$, 
\item the Sawada-Kotera equation~\cite{sk74} $u_t=u_5+5u_0u_3+5u_1u_2+5u_0^2u_1$
(which is sometimes called the Caudrey-Dodd-Gibbon equation~\cite{cdg76}).
\end{itemize}
Many more examples of integrable PDEs of this type 
can be found in~\cite{mesh-sok13,mikh91} and references therein.

Equations of the form~\er{utukd} are considered in Theorem~\ref{kntkdvtypeth}, 
which is proved in~\cite{sint18}.
Theorem~\ref{kntkdvtypeth} implies that, 
for any such equation with $\kd\in\{1,2,3\}$, 
\begin{itemize}
\item for every $\ocs\ge\kd+\delta_{\kd,3}$ the algebra 
$\fds^\ocs(\CE,a)$ is obtained from $\fds^{\ocs-1}(\CE,a)$ 
by central extension, 
\item for every $\ocs\ge\kd+\delta_{\kd,3}$ the algebra $\fds^\ocs(\CE,a)$ 
is obtained from $\fds^{\kd-1+\delta_{\kd,3}}(\CE,a)$ 
by applying several times the operation of central extension. 
\end{itemize}
Here $\delta_{\kd,3}$ is the Kronecker delta. 
So $\delta_{3,3}=1$, and $\delta_{\kd,3}=0$ if $\kd\neq 3$.

Applications of Theorem~\ref{kntkdvtypeth} to obtaining some necessary conditions 
for integrability of equations~\er{utukd} are described in~\cite{sint18}.
Results similar to Theorem~\ref{kntkdvtypeth} can be proved for many 
other evolution PDEs as well. For instance, in~\cite{sbt18} we present 
a similar result for the Krichever-Novikov equation.

Let $\bl$, $\bl_1$, $\bl_2$ be Lie algebras. 
One says that \emph{$\bl_1$ is obtained from $\bl$ by central extension} 
if there is an ideal $\mathfrak{I}\subset\bl_1$ such that 
$\mathfrak{I}$ is contained in the center of $\bl_1$ and $\bl_1/\mathfrak{I}\cong\bl$. 
Note that $\mathfrak{I}$ may be of arbitrary dimension. 

We say that \emph{$\bl_2$ is obtained from $\bl$ 
by applying several times the operation of central extension} 
if there is a finite collection of Lie algebras $\mg_0,\mg_1,\dots,\mg_k$ such 
that $\mg_0\cong\bl$, $\mg_k\cong\bl_2$ and 
$\mg_i$ is obtained from $\mg_{i-1}$ by central extension for each  
$i=1,\dots,k$. 

Consider the infinite-dimensional Lie algebra 
$\msl_2(\fik[\la])\cong \msl_2(\fik)\otimes_{\fik}\fik[\lambda]$, 
where $\fik[\lambda]$ is the algebra of polynomials in $\la$.
(If we regard $\fik$ as a rational algebraic curve with coordinate~$\la$, 
the elements of~$\msl_2(\fik[\la])$ can be identified with polynomial  
$\msl_2(\fik)$-valued functions on this rational curve.)
For the KdV equation, in Lemma~\ref{lf0kdv} we prove that 
$\fds^0(\CE,a)$ is isomorphic to the direct sum 
of $\msl_2(\fik[\la])$ and a $3$-dimensional abelian Lie algebra.

To obtain this result, we use the following fact. 
If the function $F$ in~\er{eveq_intr} does not depend on $x$, $t$, 
then the algebra $\fds^{0}(\CE,a)$ is isomorphic to a certain subalgebra of 
the Wahlquist-Estabrook prolongation algebra for~\er{eveq_intr}
(see Theorem~\ref{thmhfd0} for details).

The explicit structure of the Wahlquist-Estabrook prolongation algebra 
for the KdV equation 
is given in~\cite{kdv,kdv1} and contains $\msl_2(\fik[\la])$.
This helps us to describe $\fds^0(\CE,a)$ for KdV in Lemma~\ref{lf0kdv}.
Then Theorem~\ref{kntkdvtypeth} implies that, for every $\ocs\in\zsp$, 
the algebra $\fds^\ocs(\CE,a)$ for KdV  
is obtained from $\msl_2(\fik[\la])$ by applying several times 
the operation of central extension. See Theorem~\ref{thkdvm} for more details.

For the Krichever-Novikov equation, 
in~\cite{sbt18} we show that some infinite-dimensional Lie algebra 
of certain matrix-valued functions on an elliptic curve, 
which arises from the elliptic ZCR~\cite{krich80,novikov99} of this equation, 
plays the main role in the description of $\fds^\ocs(\CE,a)$.

Somewhat similar (but not the same) ideas on ZCRs and 
B\"acklund transformations were considered by one of us in~\cite{cfa}, 
mostly for a few scalar evolution PDEs of order~$3$.
As we show in~\cite{sint18,sbt18}, the theory of this paper 
has more applications than that of~\cite{cfa}.


For the Burgers and KdV equations, ZCRs of the form
\beq
\lb{covscal}
A=A(u_0,u_1,u_2,\dots),\qquad B=B(u_0,u_1,u_2,\dots),\qquad D_x(B)-D_t(A)+[A,B]=0 
\ee
(where $A$ and $B$ may depend on any finite number of the coordinates $u_k$) 
were studied in~\cite{finley93}. 
However, gauge transformations were not considered in~\cite{finley93}. 
Because of this, the paper~\cite{finley93} had
to impose some additional constraints on the functions $A$, $B$ in~\er{covscal}. 

\section{ZCRs, gauge transformations, and the algebras $\fds^\ocs(\CE,a)$}
\lb{csev}

Recall that $x$, $t$, $u_k$ take values in $\fik$, 
where $\fik$ is either $\Com$ or $\mathbb{R}$.
Let $\fik^\infty$ be the infinite-dimensional space  
with the coordinates $x$, $t$, $u_k$ for $k\in\zp$. 
The topology on $\fik^\infty$ is defined as follows. 

For each $l\in\zp$, consider the space $\fik^{l+3}$ 
with the coordinates $x$, $t$, $u_k$ for $k\le l$. 
One has the natural projection $\pi_l\cl\fik^\infty\to\fik^{l+3}$ that ``forgets'' 
the coordinates $u_{k'}$ for $k'>l$. 

Since $\fik^{l+3}$ is a finite-dimensional vector space, 
we have the standard topology on~$\fik^{l+3}$. 
For any $l\in\zp$ and any open subset $V\subset\fik^{l+3}$, 
the subset~$\pi_l^{-1}(V)\subset\fik^\infty$ 
is, by definition, open in $\fik^\infty$. 
Such subsets form a base of the topology on~$\fik^\infty$. 
In other words, we consider the smallest topology on~$\fik^\infty$ such that 
the maps $\pi_l$, $l\in\zp$, are continuous. 

Let $\ost\subset\fik^{\eo+3}$ be an open subset such that the function 
$F(x,t,u_0,u_1,\dots,u_{\eo})$ from~\er{eveq_intr} is defined on~$\ost$. 
The \emph{infinite prolongation} $\CE$ of equation~\er{eveq_intr} 
is defined as follows $\CE=\pi_{\eo}^{-1}(\ost)\subset\fik^\infty$.
So $\CE$ is an open subset of the space $\fik^\infty$  
with the coordinates $x$, $t$, $u_k$ for $k\in\zp$. 
The topology on~$\CE$ is induced by the embedding $\CE\subset\fik^\infty$. 

A point $a\in\CE$ is determined by the values of $x$, $t$, $u_k$ at $a$. Let
\begin{equation}
\lb{pointevs}
a=(x=x_a,\,t=t_a,\,u_k=a_k)\,\in\,\CE,\qquad\qquad x_a,\,t_a,\,a_k\in\fik,\qquad k\in\zp,
\end{equation}
be a point of $\CE$.
The constants $x_a$, $t_a$, $a_k$ are the coordinates 
of the point $a\in\CE$ in the coordinate system $x$, $t$, $u_k$.

We continue to use the notations introduced in Section~\ref{subsint1}.
In particular, $\mg\subset\gl_\sm$ is a matrix Lie algebra, 
and $\mathcal{G}\subset\mathrm{GL}_\sm$ is the connected matrix Lie group 
corresponding to $\mg$, where $\sm\in\zsp$.

For any $l\in\zp$, a matrix-function $G=G(x,t,u_0,u_1,\dots,u_l)$ 
with values in~$\mathcal{G}$ is called a \emph{gauge transformation}.
Equivalently, one can say that a gauge transformation is given 
by a $\mathcal{G}$-valued function $G=G(x,t,u_0,\dots,u_l)$.
See also Remark~\ref{rogt} about gauge transformations with values 
in other matrix Lie groups.

In this section, when we speak about ZCRs, we always 
mean ZCRs of equation~\er{eveq_intr}.
For each $i=1,2$, let 
\beq
\notag
A_i=A_i(x,t,u_0,u_1,\dots),\quad 
B_i=B_i(x,t,u_0,u_1,\dots),\quad
D_x(B_i)-D_t(A_i)+[A_i,B_i]=0
\ee
be a $\mg$-valued ZCR.
The ZCR $A_1,B_1$ is said to be \emph{gauge equivalent} 
to the ZCR $A_2,B_2$ if there is a gauge transformation $G=G(x,t,u_0,\dots,u_l)$ 
such that 
$$
A_1=GA_2G^{-1}-D_x(G)\cdot G^{-1},\qquad\qquad
B_1=GB_2G^{-1}-D_t(G)\cdot G^{-1}.
$$

The following lemma is known, but for completeness we present a proof of it.
\begin{lemma}
\lb{lemgt}
Let 
\beq
\lb{lemab}
A=A(x,t,u_0,u_1,\dots,u_\oc),\quad 
B=B(x,t,u_0,u_1,\dots,u_{\oc+\eo-1}),\quad
D_x(B)-D_t(A)+[A,B]=0
\ee
be a ZCR of order $\le\oc$ for some $\oc\in\zp$ such that 
the functions $A$, $B$ take values in $\mg$.
Here $D_x$ and $D_t$ are given by~\er{evdxdt}.

Then for any $\mathcal{G}$-valued function 
\beq
\lb{ggocs1}
G=G(x,t,u_0,u_1,\dots,u_{\oc-1})
\ee
depending on $x$, $t$, $u_0,\dots,u_{\oc-1}$, the functions 
\beq
\lb{lemtab}
\tilde{A}=GAG^{-1}-D_x(G)\cdot G^{-1},\qquad\qquad
\tilde{B}=GBG^{-1}-D_t(G)\cdot G^{-1}
\ee
form a $\mg$-valued ZCR of order $\le\oc$. 
That is, 
\beq
\lb{lemtzcr}
\tilde{A}=\tilde{A}(x,t,u_0,u_1,\dots,u_\oc),\quad 
\tilde{B}=\tilde{B}(x,t,u_0,u_1,\dots,u_{\oc+\eo-1}),\quad
D_x(\tilde{B})-D_t(\tilde{A})+[\tilde{A},\tilde{B}]=0,
\ee
and $\tilde{A}$, $\tilde{B}$ take values in $\mg$.
Formulas~\er{lemtab} determine an action of the group 
of $\mathcal{G}$-valued gauge transformations~\er{ggocs1}
on the set of $\mg$-valued ZCRs of order~$\le\oc$.
\end{lemma}
\begin{proof}
Since $A$, $B$ take values in $\mg$ and $G$ takes values in 
the connected Lie group $\mathcal{G}\subset\mathrm{GL}_\sm$
corresponding to the Lie algebra $\mg\subset\gl_\sm$,
the functions 
\beq
\lb{gagb}
GAG^{-1},\quad GBG^{-1},\quad\frac{\pd}{\pd x}(G)\cdot G^{-1},\quad
\frac{\pd}{\pd t}(G)\cdot G^{-1},\quad\frac{\pd}{\pd u_k}(G)\cdot G^{-1}\qquad
\forall\,k
\ee
take values in $\mg$. 
Hence the functions $\tilde{A}$, $\tilde{B}$ given by~\er{lemtab} 
take values in $\mg$ as well.
Using formulas~\er{evdxdt}, \er{lemab}, \er{lemtab} 
and the fact that $G$ may depend only on $x$, $t$, $u_0,\dots,u_{\oc-1}$, 
we easily get~\er{lemtzcr}. 

One has $D_x+\tilde{A}=G(D_x+A)G^{-1}$ and $D_x+\tilde{B}=G(D_t+B)G^{-1}$, 
which implies that formulas~\er{lemtab} determine an action of the group 
of $\mathcal{G}$-valued gauge transformations~\er{ggocs1}
on the set of $\mg$-valued ZCRs of order~$\le\oc$.
\end{proof}

\begin{remark}
\lb{anmer}
For any $l\in\zp$, 
when we consider a function $Q=Q(x,t,u_0,u_1,\dots,u_l)$ 
defined on a neighborhood of $a\in\CE$, 
we always assume that the function is analytic on this neighborhood.
For example, $Q$ may be a meromorphic function defined on an open subset of~$\CE$
such that $Q$ is analytic on a neighborhood of $a\in\CE$.
In particular, this applies to the functions $A$, $B$ 
considered in Theorem~\ref{evcov} below.
\end{remark}

Let $s\in\zp$. For a function $M=M(x,t,u_0,u_1,u_2,\dots)$, the notation 
$M\,\Big|_{u_k=a_k,\ k\ge s}$ 
means that we substitute $u_k=a_k$ for all $k\ge s$ in the function $M$. 
Also, sometimes we substitute $x=x_a$ or $t=t_a$ in such functions. 
For example, if $M=M(x,t,u_0,u_1,u_2,u_3)$, then 
$$
M\,\Big|_{x=x_a,\ u_k=a_k,\ k\ge 2}=M(x_a,t,u_0,u_1,a_2,a_3).
$$

\begin{theorem}
\lb{evcov}
Let $\mg\subset\gl_\sm$ be a matrix Lie algebra and  
$\mathcal{G}\subset\mathrm{GL}_\sm$ be 
the connected matrix Lie group corresponding to~$\mg$, 
where $\sm\in\zsp$. Let 
\beq
\lb{thmnoc}
A=A(x,t,u_0,u_1,\dots,u_\ocs),\quad 
B=B(x,t,u_0,u_1,\dots,u_{\ocs+\eo-1}),\quad
D_x(B)-D_t(A)+[A,B]=0
\ee
be a ZCR of order $\le\ocs$ for some $\ocs\in\zp$ 
such that the functions $A$, $B$ are defined 
on a neighborhood of $a\in\CE$ and take values in~$\mg$.

Then there is a $\mathcal{G}$-valued function $G=G(x,t,u_0,u_1,\dots,u_{\ocs-1})$ 
on a neighborhood of $a\in\CE$ such that the functions 
\beq
\lb{mnprth}
\tilde{A}=GAG^{-1}-D_x(G)\cdot G^{-1},\qquad\qquad
\tilde{B}=GBG^{-1}-D_t(G)\cdot G^{-1}
\ee
satisfy 
\begin{gather}
\label{d=0}
\frac{\pd \tilde{A}}{\pd u_s}\,\,\bigg|_{u_k=a_k,\ k\ge s}=0\qquad\quad
\forall\,s\ge 1,\\
\lb{aukak}
\tilde{A}\,\Big|_{u_k=a_k,\ k\ge 0}=0,\\
\lb{bxx0}
\tilde{B}\,\Big|_{x=x_a,\ u_k=a_k,\ k\ge 0}=0,
\end{gather}
and one has
\beq
\lb{gaid}
G\,\Big|_{x=x_a,\ t=t_a,\ u_k=a_k,\ k\ge 0}=\mathrm{Id}.
\ee

Note that, according to Lemma~\ref{lemgt}, the functions~\er{mnprth}
form a $\mg$-valued ZCR of order~$\le\ocs$. That is, 
\begin{gather}
\lb{tatbth}
\tilde{A}=\tilde{A}(x,t,u_0,u_1,\dots,u_\ocs),\quad\qquad 
\tilde{B}=\tilde{B}(x,t,u_0,u_1,\dots,u_{\ocs+\eo-1}),\\
\lb{zcrtmn}
D_x(\tilde{B})-D_t(\tilde{A})+[\tilde{A},\tilde{B}]=0,
\end{gather}
and $\tilde{A}$, $\tilde{B}$ take values in $\mg$.
Furthermore, in Theorem~\ref{tguniq} below we will show 
that a $\mathcal{G}$-valued function $G$
satisfying the above properties is unique. 
\end{theorem}
\begin{proof}
To explain the main idea, let us consider first the case $\ocs=2$. 
So $A=A(x,t,u_0,u_1,u_2)$. 

Consider the ordinary differential equation (ODE)
\beq
\lb{pdgquq}
\frac{\pd G_1}{\pd u_{1}}=G_1\cdot
\bigg(
\frac{\pd A}{\pd u_{2}}\,\,\bigg|_{u_k=a_k,\ k\ge 2}\bigg)
\ee
with respect to the variable $u_{1}$ and an unknown function   
$G_1=G_1(x,t,u_0,u_1)$. The variables $x,t,u_0$ 
are regarded as parameters in this ODE. 

Let $G_1(x,t,u_0,u_1)$ be a local 
solution of the ODE~\er{pdgquq} with the initial condition
$G_1(x,t,u_0,a_1)=\mathrm{Id}$. 
Since $\pd A/\pd u_{2}$ takes values in $\mg$, 
the function $G_1(x,t,u_0,u_1)$ takes values in $\mathcal{G}$.
Set 
\beq
\lb{a11}
\hat{A}=G_1AG_1^{-1}-D_x(G_1)\cdot G_1^{-1},\qquad\qquad
\hat{B}=G_1BG_1^{-1}-D_t(G_1)\cdot G_1^{-1}.
\ee
As $G_1$ takes values in $\mathcal{G}$, 
the functions $\hat{A}$, $\hat{B}$ take values in $\mg$.
Using~\er{a11} and~\er{pdgquq}, we get  
\begin{multline}
\lb{multgag}
\frac{\pd \hat{A}}{\pd u_2}\,\,\bigg|_{u_k=a_k,\ k\ge 2}=
G_1\bigg(
\frac{\pd A}{\pd u_{2}}\,\,\bigg|_{u_k=a_k,\ k\ge 2}\bigg)G_1^{-1}
-\bigg(
\frac{\pd}{\pd u_{2}}\big(D_x(G_1)\big)\,\bigg|_{u_k=a_k,\ k\ge 2}\bigg)G_1^{-1}=\\
=G_1\bigg(
\frac{\pd A}{\pd u_{2}}\,\bigg|_{u_k=a_k,\ k\ge 2}\bigg)G_1^{-1}
-\frac{\pd G_1}{\pd u_{1}}G_1^{-1}=
G_1\bigg(\frac{\pd A}{\pd u_{2}}\,\bigg|_{u_k=a_k,\ k\ge 2}\bigg)G_1^{-1}
-G_1\bigg(\frac{\pd A}{\pd u_{2}}\,\bigg|_{u_k=a_k,\ k\ge 2}\bigg)G_1^{-1}=0.
\end{multline}

Now consider the ODE 
\beq
\lb{pdgquq0}
\frac{\pd G_0}{\pd u_{0}}=G_0\cdot
\bigg(
\frac{\pd\hat{A}}{\pd u_{1}}\,\,\bigg|_{u_k=a_k,\ k\ge 1}\bigg)
\ee
with respect to the variable $u_{0}$ and an unknown function   
$G_0=G_0(x,t,u_0)$, where $x,\,t$ are regarded as parameters. 

Let $G_0(x,t,u_0)$ be a local 
solution of the ODE~\er{pdgquq0} with the initial condition
$G_0(x,t,a_0)=\mathrm{Id}$. 
Since ${\pd\hat{A}}/{\pd u_{1}}$ takes values in $\mg$, 
the function $G_0(x,t,u_0)$ takes values in $\mathcal{G}$.
Set 
\beq
\lb{a22}
\bar{A}=G_0\hat{A}G_0^{-1}-D_x(G_0)\cdot G_0^{-1},\qquad\qquad
\bar{B}=G_0\hat{B}G_0^{-1}-D_t(G_0)\cdot G_0^{-1}.
\ee
Then~\er{multgag},~\er{pdgquq0},~\er{a22} yield
$$
\frac{\pd \bar{A}}{\pd u_s}\,\,\bigg|_{u_k=a_k,\ k\ge s}=0\quad\qquad
\forall\,s\ge 1.
$$
Furthermore, as $G_0$ takes values in~$\mathcal{G}$, 
the functions $\bar{A}$, $\bar{B}$ take values in~$\mg$.

Let $\tilde G=\tilde G(x,t)$ be a local solution of the ODE 
\beq
\notag
\frac{\pd \tilde G}{\pd x}=\tilde G\cdot
\Big(\bar{A}\,\Big|_{u_k=a_k,\ k\ge 0}\Big)
\ee
with the initial condition $\tilde G(x_a,t)=\mathrm{Id}$, 
where $t$ is viewed as a parameter. Set 
\beq
\lb{a33}
\check{A}=\tilde{G}\bar{A}\tilde{G}^{-1}-D_x(\tilde{G})\cdot \tilde{G}^{-1},
\qquad\qquad
\check{B}=\tilde{G}\bar{B}\tilde{G}^{-1}-D_t(\tilde{G})\cdot \tilde{G}^{-1}.
\ee
Then 
$$
\frac{\pd\check{A}}{\pd u_s}\,\,\bigg|_{u_k=a_k,\ k\ge s}=0\qquad
\forall\,s\ge 1,\qquad\qquad
\check{A}\,\Big|_{u_k=a_k,\ k\ge 0}=0.
$$

Now let $\hat G=\hat G(t)$ be a local solution of the ODE 
\beq
\lb{pdhatgt}
\frac{\pd \hat G}{\pd t}=\hat G\cdot
\Big(\check{B}\,\Big|_{x=x_a,\ u_k=a_k,\ k\ge 0}\Big)
\ee
with the initial condition $\hat G(t_a)=\mathrm{Id}$.
Since $\bar{A}$ takes values in~$\mg$, the function $\tilde G$ 
takes values in $\mathcal{G}$.
Then we see that $\check{A}$, $\check{B}$ given by~\er{a33} 
take values in~$\mg$, which implies that $\hat G$ 
takes values in $\mathcal{G}$. 
Set 
\beq
\lb{a44}
\tilde{A}=\hat{G}\check{A}\hat{G}^{-1}-D_x(\hat{G})\cdot \hat{G}^{-1},
\qquad\qquad
\tilde{B}=\hat{G}\check{B}\hat{G}^{-1}-D_t(\hat{G})\cdot \hat{G}^{-1}.
\ee
Then $\tilde{A}$, $\tilde{B}$ obey~\er{d=0},~\er{aukak},~\er{bxx0} 
and take values in~$\mg$. 

Let $G=\hat G\cdot\tilde G\cdot G_0\cdot G_1$. Then 
equations~\er{a11}, \er{a22}, \er{a33}, \er{a44} imply 
\beq
\notag
\tilde{A}=GAG^{-1}-D_x(G)\cdot G^{-1},\qquad\quad
\tilde{B}=GBG^{-1}-D_t(G)\cdot G^{-1}.
\ee
Furthermore, since 
$$
G_1(x,t,u_0,a_1)=G_0(x,t,a_0)=\tilde G(x_a,t)=\hat G(t_a)=\mathrm{Id},
$$
we have $G(x_a,t_a,a^1_0,a^2_0)=\mathrm{Id}$.
Thus $G=\hat G\cdot\tilde G\cdot G_0\cdot G_1$ satisfies 
all the required properties in the case $\ocs=2$. 

This construction can be easily generalized to the case of arbitrary~$\ocs$. 
One can define $G$ as the product 
$G=\hat G\cdot\tilde G\cdot G_0\cdot G_1\dots G_{\ocs-1}$, 
where the $\mathcal{G}$-valued functions
\begin{gather*}
G_q=G_q(x,t,u_0,\dots,u_q),\qquad q=0,1,\dots,\ocs-1,\qquad\quad 
\tilde G=\tilde G(x,t),\qquad\quad\hat G=\hat G(t)
\end{gather*}
are defined as solutions of certain ODEs similar to the ODEs considered above. 
\end{proof}

Fix a point $a\in\CE$ given by~\er{pointevs}, 
which is determined by constants $x_a$, $t_a$, $a_k$.

A ZCR 
\beq
\lb{anzcr}
\anA=\anA(x,t,u_0,u_1,\dots),\qquad 
\anB=\anB(x,t,u_0,u_1,\dots),\qquad
D_x(\anB)-D_t(\anA)+[\anA,\anB]=0
\ee
is said to be \emph{$a$-normal} if $\anA$, $\anB$ satisfy the following equations
\begin{gather}
\label{agd=0}
\frac{\pd \anA}{\pd u_s}\,\,\bigg|_{u_k=a_k,\ k\ge s}=0\qquad\quad
\forall\,s\ge 1,\\
\lb{agaukak}
\anA\,\Big|_{u_k=a_k,\ k\ge 0}=0,\\
\lb{agbxx0}
\anB\,\Big|_{x=x_a,\ u_k=a_k,\ k\ge 0}=0.
\end{gather}

\begin{remark}
\lb{rnfzcr}
For example, the ZCR $\tilde{A},\tilde{B}$ 
described in Theorem~\ref{evcov} is 
$a$-normal, because $\tilde{A}$, $\tilde{B}$ obey \er{d=0}, 
\er{aukak}, \er{bxx0}.
Theorem~\ref{evcov} implies that any ZCR on a neighborhood of $a\in\CE$
is gauge equivalent to an $a$-normal ZCR.
Therefore, we can say that 
properties \er{agd=0}, \er{agaukak}, \er{agbxx0} 
determine a normal form for ZCRs
with respect to the action of the group of gauge transformations 
on a neighborhood of $a\in\CE$.
\end{remark}

Analyzing properties \er{agd=0}, \er{agaukak}, \er{agbxx0} 
of $a$-normal ZCRs, it is easy to prove the following lemma.
\begin{lemma}
\lb{lange}
Let $\oc_1,\oc_2\in\zp$. 
For each $i=1,2$, let 
\beq
\notag
\anA_i=\anA_i(x,t,u_0,\dots,u_{\oc_i}),\quad 
\anB_i=\anB_i(x,t,u_0,\dots,u_{\oc_i+\eo-1}),\quad
D_x(\anB_i)-D_t(\anA_i)+[\anA_i,\anB_i]=0
\ee
be an $a$-normal ZCR of order~$\le\oc_i$ such that 
the functions $\anA_i$, $\anB_i$ are defined 
on a neighborhood of $a\in\CE$ and take values in~$\mg$.

Suppose that on a neighborhood of $a\in\CE$ 
there is a function $\bG=\bG(x,t,u_0,\dots,u_l)$ 
with values in $\mathcal{G}$ such that
\beq
\notag
\anA_1=\bG\anA_2\bG^{-1}-D_x(\bG)\cdot\bG^{-1},\qquad\qquad
\anB_1=\bG\anB_2\bG^{-1}-D_t(\bG)\cdot\bG^{-1}.
\ee
In other words, we suppose that the $a$-normal ZCR $\anA_1,\anB_1$ is 
gauge equivalent to the $a$-normal ZCR $\anA_2,\anB_2$ 
by means of a gauge transformation $\bG=\bG(x,t,u_0,\dots,u_l)$.

Then the function $\bG$ is actually a constant element of the group $\mathcal{G}$
\textup{(}that is, $\bG$ does not depend on $x$, $t$, $u_k$\textup{)}, and we have
\beq
\notag
\anA_1=\bG\anA_2\bG^{-1},\qquad\qquad
\anB_1=\bG\anB_2\bG^{-1}.
\ee
\end{lemma}

\begin{theorem}
\lb{tguniq}
We use here the notations introduced in Theorem~\ref{evcov}.
Let 
\beq
\lb{guzcr}
A=A(x,t,u_0,\dots,u_\ocs),\quad 
B=B(x,t,u_0,\dots,u_{\ocs+\eo-1}),\quad
D_x(B)-D_t(A)+[A,B]=0
\ee
be a ZCR of order~$\le\ocs$ such that 
the functions $A$, $B$ are defined 
on a neighborhood of $a\in\CE$ and take values in~$\mg$.

Then on a neighborhood of $a\in\CE$ 
there is a unique gauge transformation $G=G(x,t,u_0,\dots,u_l)$ 
such that $G(a)=\mathrm{Id}$ and the functions 
\beq
\lb{tatbgu}
\tilde{A}=GAG^{-1}-D_x(G)\cdot G^{-1},\qquad\qquad
\tilde{B}=GBG^{-1}-D_t(G)\cdot G^{-1}
\ee
form an $a$-normal ZCR.
\textup{(}That is, the functions~\er{tatbgu} satisfy 
\er{d=0}, \er{aukak}, \er{bxx0}, \er{zcrtmn}.\textup{)}
Furthermore, $G$ depends only on $x$, $t$, $u_0,\dots,u_{\oc-1}$, 
and the ZCR~\er{tatbgu} is of order~$\le\ocs$.

Note that, according to our definition of gauge transformations, 
the function $G$ takes values in $\mathcal{G}$. 
The property $G(a)=\mathrm{Id}$ means that $G(x_a,t_a,a_0,\dots,a_{\oc-1})=\mathrm{Id}$.
\end{theorem}
\begin{proof}
Existence of the required gauge transformation follows from Theorem~\ref{evcov}.
Let us prove uniqueness of it.

Suppose that we have two gauge transformations 
$$
G_1=G_1(x,t,u_0,\dots,u_{l_1}),\qquad\qquad
G_2=G_2(x,t,u_0,\dots,u_{l_2})
$$ 
such that $G_1(a)=G_2(a)=\mathrm{Id}$ and for each $i=1,2$ the ZCR 
given by the functions
\beq
\lb{anabi}
\anA_i=G_iAG_i^{-1}-D_x(G_i)\cdot G_i^{-1},\qquad\qquad
\anB_i=G_iBG_i^{-1}-D_t(G_i)\cdot G_i^{-1}
\ee
is $a$-normal. 

Relations~\er{anabi} say the following.
For each $i=1,2$, 
applying the gauge transformation $G_i$ to the ZCR $A,B$,
we get the ZCR $\anA_i,\anB_i$.
Therefore, applying the gauge transformation $G_1G_2^{-1}$ 
to the ZCR $\anA_2,\anB_2$, we get the ZCR $\anA_1,\anB_1$.
That is, from~\er{anabi} one obtains 
\begin{gather*}
\notag
\anA_1=(G_1G_2^{-1})\anA_2(G_1G_2^{-1})^{-1}
-D_x(G_1G_2^{-1})\cdot(G_1G_2^{-1})^{-1},\\
\anB_1=(G_1G_2^{-1})\anB_2(G_1G_2^{-1})^{-1}
-D_t(G_1G_2^{-1})\cdot(G_1G_2^{-1})^{-1},
\end{gather*}
which means that the $a$-normal ZCR $\anA_1,\anB_1$ is 
gauge equivalent to the $a$-normal ZCR $\anA_2,\anB_2$ 
by means of the gauge transformation $G_1G_2^{-1}$.
Then, by Lemma~\ref{lange}, the $\mathcal{G}$-valued function $G_1G_2^{-1}$ 
is a constant element of the group $\mathcal{G}$.
Since $G_1(a)=G_2(a)=\mathrm{Id}$, this implies $G_1=G_2$.
\end{proof}

\begin{remark}
\lb{abcoef0}
According to Remark~\ref{anmer}, the $\mg$-valued functions~\er{thmnoc}
are analytic on a neighborhood of $a\in\CE$.
The construction of $G=G(x,t,u_0,\dots,u_{\ocs-1})$ 
in the proof of Theorem~\ref{evcov} implies that $G$ is analytic as well.
Then the $\mg$-valued functions $\tilde{A}$, $\tilde{B}$ given by~\er{mnprth} 
are also analytic on a neighborhood of $a\in\CE$.

Since $\tilde{A}$, $\tilde{B}$ are analytic and are of the form~\er{tatbth},
these functions are represented as absolutely convergent power series
\begin{gather}
\label{aser}
\tilde{A}=\sum_{l_1,l_2,i_0,\dots,i_\ocs\ge 0} 
(x-x_a)^{l_1} (t-t_a)^{l_2}(u_0-a_0)^{i_0}\dots(u_\ocs-a_\ocs)^{i_\ocs}\cdot
\tilde{A}^{l_1,l_2}_{i_0\dots i_\ocs},\\
\lb{bser}
\tilde{B}=\sum_{l_1,l_2,j_0,\dots,j_{\ocs+\eo-1}\ge 0} 
(x-x_a)^{l_1} (t-t_a)^{l_2}(u_0-a_0)^{j_0}\dots(u_{\ocs+\eo-1}-a_{\ocs+\eo-1})^{j_{\ocs+\eo-1}}\cdot
\tilde{B}^{l_1,l_2}_{j_0\dots j_{\ocs+\eo-1}},\\
\notag
\tilde{A}^{l_1,l_2}_{i_0\dots i_\ocs},\,\tilde{B}^{l_1,l_2}_{j_0\dots j_{\ocs+\eo-1}}\in\mg.
\end{gather}

For each $k\in\zsp$, we set 
\beq
\lb{defzcsoc}
\zcs_k=\Big\{(i_0,\dots,i_{k})\in\zp^{k+1}\ \Big|\ \exists\,r\in\{1,\dots,k\}\,\ 
\text{such that}\,\ i_r=1,\,\ i_q=0\,\ \forall\,q>r\Big\}.
\ee
In other words, for $k\in\zsp$ and $i_0,\dots,i_{k}\in\zp$, one has $(i_0,\dots,i_{k})\in\zcs_k$ iff
there is $r\in\{1,\dots,k\}$ such that 
$(i_0,\dots,i_{r-1},i_r,i_{r+1},\dots,i_{k})=(i_0,\dots,i_{r-1},1,0,\dots,0)$.
Set also $\zcs_0=\varnothing$.
So the set $\zcs_0$ is empty.

Using formulas~\er{aser},~\er{bser}, we see that properties~\er{d=0},~\er{aukak}, \er{bxx0}
are equivalent to 
\beq
\lb{ab000}
\tilde{A}^{l_1,l_2}_{0\dots 0}=\tilde{B}^{0,l_2}_{0\dots 0}=0,\qquad
\tilde{A}^{l_1,l_2}_{i_0\dots i_{\ocs}}=0,
\qquad (i_0,\dots,i_{\ocs})\in\zcs_\ocs,\qquad l_1,l_2\in\zp.
\ee
\end{remark}

\begin{remark}
\lb{psdxdtlie}
Let $\bl$ be a Lie algebra. 
Consider a formal power series of the form 
$$
C=\sum_{l_1,l_2,i_0,\dots,i_m\ge 0} 
(x-x_a)^{l_1} (t-t_a)^{l_2}(u_0-a_0)^{i_0}\dots(u_m-a_m)^{i_m}\cdot 
C^{l_1,l_2}_{i_0\dots i_m},\qquad\quad
C^{l_1,l_2}_{i_0\dots i_m}\in\bl.
$$
Set 
\begin{gather}
\lb{dxck}
D_x(C)=\sum_{l_1,l_2,i_0,\dots,i_m} 
D_x\big((x-x_a)^{l_1} (t-t_a)^{l_2}(u_0-a_0)^{i_0}\dots(u_m-a_m)^{i_m}\big)\cdot 
C^{l_1,l_2}_{i_0\dots i_m},\\
\lb{dtck}
D_t(C)=\sum_{l_1,l_2,i_0,\dots,i_m} 
D_t\big((x-x_a)^{l_1} (t-t_a)^{l_2}(u_0-a_0)^{i_0}\dots(u_m-a_m)^{i_m}\big)\cdot 
C^{l_1,l_2}_{i_0\dots i_m}.
\end{gather}
The expressions 
\beq
\lb{dxtua}
\begin{aligned}
D_x\big((x-x_a)^{l_1} (t-t_a)^{l_2}(u_0-a_0)^{i_0}\dots(u_m-a_m)^{i_m}\big),\\
D_t\big((x-x_a)^{l_1} (t-t_a)^{l_2}(u_0-a_0)^{i_0}\dots(u_m-a_m)^{i_m}\big)
\end{aligned}
\ee
are functions of the variables $x$, $t$, $u_k$. 
Taking the corresponding Taylor series at the point~\eqref{pointevs}, 
we regard~\er{dxtua} as power series. 
Then~\er{dxck},~\er{dtck} become formal power series with coefficients in $\bl$. 

According to~\er{evdxdt}, one has  
$D_t=\frac{\pd}{\pd t}+\sum_{k\ge 0} D_x^k(F)\frac{\pd}{\pd u_k}$, 
where $F=F(x,t,u_0,\dots,u_{\eo})$ is given in~\er{eveq_intr}. 
When we apply $D_t$ in~\er{dtck}, 
we view $F$ as a power series, using the Taylor series of the function~$F$.

Consider another formal power series
$$
R=\sum_{q_1,q_2,j_0,\dots,j_m\ge 0} 
(x-x_a)^{q_1} (t-t_a)^{q_2}(u_0-a_0)^{j_0}\dots(u_m-a_m)^{j_m}\cdot 
R^{q_1,q_2}_{j_0\dots j_m},\qquad\quad
R^{q_1,q_2}_{j_0\dots j_m}\in\bl.
$$
Then the Lie bracket $[C,R]$ is defined as follows
$$
[C,R]=\sum_{\substack{l_1,l_2,i_0,\dots,i_m,\\ q_1,q_2,j_0,\dots,j_m}}
(x-x_a)^{l_1+q_1} (t-t_a)^{l_2+q_2}(u_0-a_0)^{i_0+j_0}\dots(u_m-a_m)^{i_m+j_m}\cdot 
\big[C^{l_1,l_2}_{i_0\dots i_m},\,R^{q_1,q_2}_{j_0\dots j_m}\big].
$$

\end{remark}

\begin{remark}  
\label{inform}
The main idea of the definition of the Lie algebra $\fds^\ocs(\CE,a)$  
can be informally outlined as follows. 
According to Theorem~\ref{evcov} and Remark~\ref{abcoef0}, 
any ZCR~\er{thmnoc} of order $\le\ocs$ is gauge equivalent 
to a ZCR given by functions $\tilde{A}$, $\tilde{B}$ 
that are of the form~\er{aser},~\er{bser} 
and satisfy~\er{zcrtmn}, \er{ab000}.

To define $\fds^\ocs(\CE,a)$, we regard $\tilde{A}^{l_1,l_2}_{i_0\dots i_\ocs}$, 
$\tilde{B}^{l_1,l_2}_{j_0\dots j_{\ocs+\eo-1}}$ from~\er{aser},~\er{bser} 
as abstract symbols. 
By definition, the algebra $\fds^\ocs(\CE,a)$ is generated by the symbols 
$\tilde{A}^{l_1,l_2}_{i_0\dots i_\ocs}$, $\tilde{B}^{l_1,l_2}_{j_0\dots j_{\ocs+\eo-1}}$
for $l_1,l_2,i_0,\dots,i_\ocs,j_0,\dots,j_{\ocs+\eo-1}\in\zp$.
Relations for these generators are provided by equations~\er{zcrtmn}, \er{ab000}.
The details of this construction are presented below. 

\end{remark}

Let $\frl$ be the free Lie algebra generated 
by the symbols $\fla^{l_1,l_2}_{i_0\dots i_\ocs}$ and $\flb^{l_1,l_2}_{j_0\dots j_{\ocs+\eo-1}}$ 
for all 
$$
l_1,l_2,i_0,\dots,i_\ocs,j_0,\dots,j_{\ocs+\eo-1}\in\zp.
$$
Consider the following power series with coefficients in~$\frl$
\begin{gather*}
\fla=\sum_{l_1,l_2,i_0,\dots,i_\ocs\ge 0} 
(x-x_a)^{l_1} (t-t_a)^{l_2}(u_0-a_0)^{i_0}\dots(u_\ocs-a_\ocs)^{i_\ocs}\cdot
\fla^{l_1,l_2}_{i_0\dots i_\ocs},\\
\flb=\sum_{l_1,l_2,j_0,\dots,j_{\ocs+\eo-1}\ge 0} 
(x-x_a)^{l_1} (t-t_a)^{l_2}(u_0-a_0)^{j_0}\dots(u_{\ocs+\eo-1}-a_{\ocs+\eo-1})^{j_{\ocs+\eo-1}}\cdot
\flb^{l_1,l_2}_{j_0\dots j_{\ocs+\eo-1}}.
\end{gather*}

Then the power series $D_x(\flb)$, $D_t(\fla)$, $[\fla,\flb]$ 
are defined according to Remark~\ref{psdxdtlie}. We have 
\begin{equation*}
D_x(\flb)-D_t(\fla)+[\fla,\flb]=
\sum_{l_1,l_2,q_0,\dots,q_{\ocs+\eo}\ge 0}
(x-x_a)^{l_1}(t-t_a)^{l_2}(u_0-a_0)^{q_0}\dots(u_{\ocs+\eo}-a_{\ocs+\eo})^{q_{\ocs+\eo}}
\cdot\flz^{l_1,l_2}_{q_0\dots q_{\ocs+\eo}}
\end{equation*}
for some elements $\flz^{l_1,l_2}_{q_0\dots q_{\ocs+\eo}}\in\frl$. 

Let $\frid\subset\frl$ be the ideal generated by the elements
\begin{gather*}
\flz^{l_1,l_2}_{q_0\dots q_{\ocs+\eo}},\qquad\fla^{l_1,l_2}_{0\dots 0},\qquad 
\flb^{0,l_2}_{0\dots 0},\qquad l_1,l_2,q_0,\dots,q_{\ocs+\eo}\in\zp,\\
\fla^{l_1,l_2}_{i_0\dots i_{\ocs}},\qquad (i_0,\dots,i_{\ocs})\in\zcs_\ocs,\qquad l_1,l_2\in\zp.
\end{gather*}
Set $\fds^\ocs(\CE,a)=\frl/\frid$. 
Consider the natural homomorphism  
$\psi\cl\frl\to\frl/\frid=\fds^\ocs(\CE,a)$ and set 
$$
\ga^{l_1,l_2}_{i_0\dots i_\ocs}=\psi\big(\fla^{l_1,l_2}_{i_0\dots i_\ocs}\big),\qquad\qquad 
\gb^{l_1,l_2}_{j_0\dots j_{\ocs+\eo-1}}=\psi\big(\flb^{l_1,l_2}_{j_0\dots j_{\ocs+\eo-1}}\big).
$$
The definition of~$\frid$ implies that the power series 
\begin{gather}
\label{gasumxt}
\ga=\sum_{l_1,l_2,i_0,\dots,i_\ocs\ge 0} 
(x-x_a)^{l_1} (t-t_a)^{l_2}(u_0-a_0)^{i_0}\dots(u_\ocs-a_\ocs)^{i_\ocs}\cdot
\ga^{l_1,l_2}_{i_0\dots i_\ocs},\\
\label{gbsumxt}
\gb=\sum_{l_1,l_2,j_0,\dots,j_{\ocs+\eo-1}\ge 0} 
(x-x_a)^{l_1} (t-t_a)^{l_2}(u_0-a_0)^{j_0}\dots(u_{\ocs+\eo-1}-a_{\ocs+\eo-1})^{j_{\ocs+\eo-1}}\cdot
\gb^{l_1,l_2}_{j_0\dots j_{\ocs+\eo-1}}
\end{gather}
satisfy 
\beq
\lb{xgbtga}
D_x(\gb)-D_t(\ga)+[\ga,\gb]=0.
\ee

\begin{remark}
\lb{rem_fdpgen}
The Lie algebra $\fds^\ocs(\CE,a)$ can be described in terms 
of generators and relations as follows. 

Equation~\er{xgbtga} is equivalent to some Lie algebraic relations for 
$\ga^{l_1,l_2}_{i_0\dots i_\ocs}$, $\gb^{l_1,l_2}_{j_0\dots j_{\ocs+\eo-1}}$.
The algebra $\fds^\ocs(\CE,a)$ is given by the generators 
$\ga^{l_1,l_2}_{i_0\dots i_\ocs}$, $\gb^{l_1,l_2}_{j_0\dots j_{\ocs+\eo-1}}$,  
the relations arising from~\er{xgbtga}, and the following relations 
\beq
\lb{gagb00}
\ga^{l_1,l_2}_{0\dots 0}=\gb^{0,l_2}_{0\dots 0}=0,\qquad 
\ga^{l_1,l_2}_{i_0\dots i_{\ocs}}=0,
\qquad (i_0,\dots,i_{\ocs})\in\zcs_\ocs,\qquad l_1,l_2\in\zp.
\ee
Note that condition~\er{gagb00} is equivalent to the following equations
\begin{gather}
\lb{pdgau}
\frac{\pd\ga}{\pd u_s}\,\,\bigg|_{u_k=a_k,\ k\ge s}=0\qquad\quad
\forall\,s\ge 1,\\
\lb{gaua0}
\ga\,\Big|_{u_k=a_k,\ k\ge 0}=0,\\
\lb{gbxua0}
\gb\,\Big|_{x=x_a,\ u_k=a_k,\ k\ge 0}=0.
\end{gather}

Note that, according to Remark~\ref{psdxdtlie}, 
the definition of the power series $D_t(\ga)$ in~\er{xgbtga}  
uses the Taylor series of the function $F=F(x,t,u_0,\dots,u_{\eo})$ from~\er{eveq_intr}, 
because $D_t$ is determined by $F$.
So the constructed generators and relations for the algebra $\fds^\ocs(\CE,a)$ 
are determined by the Taylor series of the function~$F$ at the point~\eqref{pointevs}.

\end{remark}

\begin{remark}
\lb{rfzcr}
Let $\bl$ be a Lie algebra. 
If $A$, $B$ are functions with values in $\bl$ and satisfy~\er{thmnoc} 
then $A$, $B$ constitute a ZCR of order~$\le\oc$ with values in $\bl$.

Instead of functions with values in $\bl$, 
one can consider formal power series with coefficients in~$\bl$.
Then one gets the notion of \emph{formal ZCRs with coefficients in~$\bl$}.

More precisely, a \emph{formal ZCR of order~$\le\oc$ with coefficients in $\bl$}
is given by formal power series 
\begin{gather}
\label{anasum}
\anA=\sum_{l_1,l_2,i_0,\dots,i_\ocs\ge 0} 
(x-x_a)^{l_1} (t-t_a)^{l_2}(u_0-a_0)^{i_0}\dots(u_\ocs-a_\ocs)^{i_\ocs}\cdot
\anA^{l_1,l_2}_{i_0\dots i_\ocs},\\
\label{anbsum}
\anB=
\sum_{l_1,l_2,j_0,\dots,j_{\ocs+\eo-1}\ge 0} 
(x-x_a)^{l_1} (t-t_a)^{l_2}(u_0-a_0)^{j_0}\dots(u_{\ocs+\eo-1}-a_{\ocs+\eo-1})^{j_{\ocs+\eo-1}}\cdot
\anB^{l_1,l_2}_{j_0\dots j_{\ocs+\eo-1}}
\end{gather}
such that 
$\anA^{l_1,l_2}_{i_0\dots i_\ocs},\anB^{l_1,l_2}_{j_0\dots j_{\ocs+\eo-1}}\in\bl$ and
\beq
\lb{dxanb0}
D_x(\anB)-D_t(\anA)+[\anA,\anB]=0.
\ee
If the power series~\er{anasum},~\er{anbsum} 
satisfy \er{agd=0}, \er{agaukak}, \er{agbxx0} 
then this formal ZCR is said to be \emph{$a$-normal}.

For example, since \er{gasumxt}, \er{gbsumxt} 
obey \er{xgbtga}, \er{pdgau}, \er{gaua0}, \er{gbxua0} and 
$\ga^{l_1,l_2}_{i_0\dots i_\ocs},\gb^{l_1,l_2}_{j_0\dots j_{\ocs+\eo-1}}\in\fd^\oc(\CE,a)$, 
the power series \er{gasumxt}, \er{gbsumxt} constitute an $a$-normal 
formal ZCR of order~$\le\oc$ with coefficients in $\fd^\oc(\CE,a)$.
\end{remark}

\begin{remark}
\lb{rczcrh}
Let $\mg$ be a finite-dimensional Lie algebra.
Let $\hrf\cl\fds^\ocs(\CE,a)\to\mg$ be a homomorphism 
from $\fds^\ocs(\CE,a)$ to $\mg$.
Applying $\hrf$ to the coefficients of the power series \er{gasumxt}, \er{gbsumxt}, 
we get the following power series with coefficients in~$\mg$
\begin{gather}
\label{hraser}
\anA=\sum_{l_1,l_2,i_0,\dots,i_\ocs} 
(x-x_a)^{l_1} (t-t_a)^{l_2}(u_0-a_0)^{i_0}\dots(u_\ocs-a_\ocs)^{i_\ocs}\cdot
\hrf\big(\ga^{l_1,l_2}_{i_0\dots i_\ocs}\big),\\
\lb{hrbser}
\anB=\sum_{l_1,l_2,j_0,\dots,j_{\ocs+\eo-1}} 
(x-x_a)^{l_1} (t-t_a)^{l_2}(u_0-a_0)^{j_0}\dots(u_{\ocs+\eo-1}-a_{\ocs+\eo-1})^{j_{\ocs+\eo-1}}\cdot
\hrf\big(\gb^{l_1,l_2}_{j_0\dots j_{\ocs+\eo-1}}\big).
\end{gather} 
Since \er{gasumxt}, \er{gbsumxt} 
obey \er{xgbtga}, \er{pdgau}, \er{gaua0}, \er{gbxua0}, 
the power series~\er{hraser},~\er{hrbser} 
satisfy \er{agd=0}, \er{agaukak}, \er{agbxx0}, \er{dxanb0}.
Therefore, \er{hraser},~\er{hrbser}
constitute an $a$-normal formal ZCR of order~$\le\oc$ 
with coefficients in~$\mg$.

A homomorphism $\hrf\cl\fds^\ocs(\CE,a)\to\mg$ 
is said to be \emph{regular} if the power series~\er{hraser},~\er{hrbser}
are absolutely convergent in a neighborhood of $a\in\CE$. 
In other words, $\hrf$ is regular iff \er{hraser},~\er{hrbser} 
are analytic functions with values in $\mg$ on a neighborhood of $a\in\CE$.

For a regular homomorphism $\hrf$, 
the analytic functions~\er{hraser},~\er{hrbser} 
form an $a$-normal $\mg$-valued ZCR of order~$\le\oc$.
We denote this ZCR by $\mathbf{Z}(\CE,a,\ocs,\hrf)$.
Formulas~\er{hraser},~\er{hrbser} imply that 
the ZCR $\mathbf{Z}(\CE,a,\ocs,\hrf)$ takes values in 
the Lie subalgebra $\hrf\big(\fds^\ocs(\CE,a)\big)\subset\mg$.
\end{remark}

\begin{remark}
\lb{rcrhzcr}
Let $\mg$ be a finite-dimensional matrix Lie algebra.
By Theorems~\ref{evcov},~\ref{tguniq}, 
for any $\mg$-valued ZCR~\er{thmnoc} of order~$\le\ocs$ 
on a neighborhood of $a\in\CE$, 
there is a unique gauge transformation $G$ 
such that $G(a)=\mathrm{Id}$ and the functions~\er{mnprth} 
form an $a$-normal ZCR.
(That is, the functions~\er{mnprth} satisfy 
\er{d=0}, \er{aukak}, \er{bxx0}, \er{zcrtmn}.)

Consider the Taylor series~\er{aser},~\er{bser} 
of the functions~\er{mnprth}.
Properties~\er{zcrtmn},~\er{ab000} imply that the following homomorphism 
\beq
\lb{hrfdef}
\hrf\cl\fds^\ocs(\CE,a)\to\mg,\qquad
\hrf\big(\ga^{l_1,l_2}_{i_0\dots i_\ocs}\big)=
\tilde{A}^{l_1,l_2}_{i_0\dots i_\ocs},\qquad
\hrf\big(\gb^{l_1,l_2}_{j_0\dots j_{\ocs+\eo-1}}\big)=
\tilde{B}^{l_1,l_2}_{j_0\dots j_{\ocs+\eo-1}},
\ee
is well defined.
Here 
$\tilde{A}^{l_1,l_2}_{i_0\dots i_\ocs},
\tilde{B}^{l_1,l_2}_{j_0\dots j_{\ocs+\eo-1}}\in\mg$ 
are the coefficients of the power series~\er{aser},~\er{bser}.

Since \er{aser},~\er{bser} are the Taylor series of the 
analytic functions~\er{mnprth}, the homomorphism~\er{hrfdef} is regular.
According to Remark~\ref{rczcrh},
we get also the $\mg$-valued ZCR $\mathbf{Z}(\CE,a,\ocs,\hrf)$ 
corresponding to the regular homomorphism~\er{hrfdef}.
The ZCR $\mathbf{Z}(\CE,a,\ocs,\hrf)$ 
coincides with the ZCR given by the functions~\er{mnprth}.

Since the ZCR~\er{thmnoc} is gauge equivalent to the ZCR
given by~\er{mnprth},
we see that the ZCR~\er{thmnoc} is gauge equivalent 
to the ZCR $\mathbf{Z}(\CE,a,\ocs,\hrf)$.

\end{remark}

\begin{theorem}
\lb{thzcrfd}
Let $\mg\subset\gl_\sm$ be a matrix Lie algebra and  
$\mathcal{G}\subset\mathrm{GL}_\sm$ be 
the connected matrix Lie group corresponding to~$\mg$,
where $\sm\in\zsp$. 
In what follows, all ZCRs are defined on a neighborhood of $a\in\CE$.
Let $\oc\in\zp$. Consider $\mg$-valued ZCRs of order~$\le\ocs$
\beq
\lb{tcaboc}
A=A(x,t,u_0,\dots,u_\ocs),\quad 
B=B(x,t,u_0,\dots,u_{\ocs+\eo-1}),\quad
D_x(B)-D_t(A)+[A,B]=0.
\ee
We have the following correspondence between $\mg$-valued ZCRs~\er{tcaboc}
and homomorphisms $\hrf\cl\fds^\ocs(\CE,a)\to\mg$.
\begin{itemize}
\item In Remark~\ref{rcrhzcr}, for any $\mg$-valued ZCR~\er{tcaboc}, 
we have canonically defined a regular homomorphism $\hrf\cl\fds^\ocs(\CE,a)\to\mg$, 
so that the ZCR~\er{tcaboc} is gauge equivalent 
to the ZCR $\mathbf{Z}(\CE,a,\ocs,\hrf)$ defined in Remark~\ref{rczcrh}.
The ZCR $\mathbf{Z}(\CE,a,\ocs,\hrf)$ is $a$-normal 
and takes values in the Lie subalgebra $\hrf\big(\fds^\ocs(\CE,a)\big)\subset\mg$.

\item In Remark~\ref{rczcrh}, for any homomorphism 
${\hrf}\cl\fds^\ocs(\CE,a)\to\mg$, 
we have canonically defined a formal ZCR of order~$\le\oc$ 
with coefficients in~$\mg$.
This formal ZCR is given by the formal 
power series~\er{hraser},~\er{hrbser} and is $a$-normal.
If the homomorphism ${\hrf}$ is regular, then this formal ZCR is analytic 
and coincides with the ZCR $\mathbf{Z}(\CE,a,\ocs,{\hrf})$.

\item
For each $i=1,2$, let 
\beq
\lb{aibit}
A_i=A_i(x,t,u_0,\dots,u_{\oc}),\quad 
B_i=B_i(x,t,u_0,\dots,u_{\oc+\eo-1}),\quad
D_x(B_i)-D_t(A_i)+[A_i,B_i]=0
\ee
be a $\mg$-valued ZCR of order~$\le\oc$.
Let $\hrf_i\cl\fds^\ocs(\CE,a)\to\mg$ be the regular homomorphism 
associated with the ZCR~\er{aibit} by the construction in Remark~\ref{rcrhzcr}.

Then we have the following property.
The ZCR $A_1,B_1$ is gauge equivalent to the ZCR $A_2,B_2$ 
iff there is an element $\bG\in\mathcal{G}$ such that
\beq
\lb{hrfvbg}
\hrf_1(v)=\bG\cdot\hrf_2(v)\cdot\bG^{-1}
\quad\qquad\forall\,v\in\fds^{\ocs}(\CE,a).
\ee
\end{itemize}
\end{theorem}
\begin{proof}
We need to prove only the last statement of the theorem, 
because the other statements follow from Remarks~\ref{rczcrh},~\ref{rcrhzcr}.

According to Remark~\ref{rczcrh},
for each $i=1,2$ the ZCR $A_i,B_i$ is gauge equivalent to the 
$a$-normal ZCR $\mathbf{Z}(\CE,a,\ocs,\hrf_i)$.
Therefore, the ZCR $A_1,B_1$ is gauge equivalent to the ZCR $A_2,B_2$ 
iff $\mathbf{Z}(\CE,a,\ocs,\hrf_1)$ is gauge equivalent to
$\mathbf{Z}(\CE,a,\ocs,\hrf_2)$.

If there is an element $\bG\in\mathcal{G}$ satisfying~\er{hrfvbg}, 
then $\mathbf{Z}(\CE,a,\ocs,\hrf_1)$ is gauge equivalent to
$\mathbf{Z}(\CE,a,\ocs,\hrf_2)$ by means of the constant gauge transformation 
equal to $\bG$.

Conversely, if $\mathbf{Z}(\CE,a,\ocs,\hrf_1)$ 
is gauge equivalent to $\mathbf{Z}(\CE,a,\ocs,\hrf_2)$ 
by means of some gauge transformation, 
then existence of an element $\bG\in\mathcal{G}$ satisfying~\er{hrfvbg} 
follows from Lemma~\ref{lange},
because the ZCRs $\mathbf{Z}(\CE,a,\ocs,\hrf_1)$ and 
$\mathbf{Z}(\CE,a,\ocs,\hrf_2)$ are $a$-normal.
Indeed, by Lemma~\ref{lange}, if $\mathbf{Z}(\CE,a,\ocs,\hrf_1)$ 
is gauge equivalent to $\mathbf{Z}(\CE,a,\ocs,\hrf_2)$
by means of some gauge transformation,  
then this gauge transformation is actually 
a constant element $\bG\in\mathcal{G}$ obeying
\beq
\lb{zzbg}
\mathbf{Z}(\CE,a,\ocs,\hrf_1)=
\bG\cdot\mathbf{Z}(\CE,a,\ocs,\hrf_2)\cdot\bG^{-1}.
\ee
The definition of $\mathbf{Z}(\CE,a,\ocs,\hrf)$ 
in Remark~\ref{rczcrh} implies that \er{zzbg} is equivalent to~\er{hrfvbg}.
\end{proof}

\begin{remark}
\lb{rzcrcr}
Since we assume $\mg\subset\gl_\sm$ for some $\sm\in\zsp$, 
homomorphisms $\hrf\cl\fds^{\ocs}(\CE,a)\to\mg$ are representations 
of the Lie algebra $\fds^{\ocs}(\CE,a)$.
So from Theorem~\ref{thzcrfd} we see that 
$\mg$-valued ZCRs of order~$\le\ocs$ are classified 
by $\mg$-valued representations of $\fds^{\ocs}(\CE,a)$.
\end{remark}

Suppose that $\ocs\ge 1$. 
According to Remark~\ref{rem_fdpgen}, 
the algebra $\fds^\ocs(\CE,a)$ is given by the generators 
$\ga^{l_1,l_2}_{i_0\dots i_\ocs}$, $\gb^{l_1,l_2}_{j_0\dots j_{\ocs+\eo-1}}$
and the relations arising from~\er{xgbtga},~\er{gagb00}. 
Similarly, the algebra $\fds^{\ocs-1}(\CE,a)$ is given by the generators 
$\hat\ga^{l_1,l_2}_{i_0\dots i_{\ocs-1}}$, $\hat\gb^{l_1,l_2}_{j_0\dots j_{\ocs+\eo-2}}$
and the relations arising from 
\begin{gather*}
D_x\big(\hat\gb\big)-D_t\big(\hat\ga\big)+\big[\hat\ga,\hat\gb\big]=0,\\
\hat\ga^{l_1,l_2}_{0\dots 0}=\hat\gb^{0,l_2}_{0\dots 0}=0,\qquad 
\hat\ga^{l_1,l_2}_{i_0\dots i_{\ocs-1}}=0,
\qquad (i_0,\dots,i_{\ocs-1})\in\zcs_{\ocs-1},\qquad l_1,l_2\in\zp,
\end{gather*}
where 
\begin{gather*}
\hat\ga=\sum_{l_1,l_2,i_0,\dots,i_{\ocs-1}} 
(x-x_a)^{l_1} (t-t_a)^{l_2}(u_0-a_0)^{i_0}\dots(u_{\ocs-1}-a_{\ocs-1})^{i_{\ocs-1}}\cdot
\hat\ga^{l_1,l_2}_{i_0\dots i_{\ocs-1}},\\
\hat\gb=\sum_{l_1,l_2,j_0,\dots,j_{\ocs+\eo-2}} 
(x-x_a)^{l_1} (t-t_a)^{l_2}(u_0-a_0)^{j_0}\dots(u_{\ocs+\eo-2}-a_{\ocs+\eo-2})^{j_{\ocs+\eo-2}}
\cdot\hat\gb^{l_1,l_2}_{j_0\dots j_{\ocs+\eo-2}}.
\end{gather*}

This implies that the map 
\beq
\lb{fdhff}
\ga^{l_1,l_2}_{i_0\dots i_{\ocs-1}i_{\ocs}}\,\mapsto\,
\delta_{0,i_{\ocs}}\cdot\hat\ga^{l_1,l_2}_{i_0\dots i_{\ocs-1}},\qquad\quad 
\gb^{l_1,l_2}_{j_0\dots j_{\ocs+\eo-2}j_{\ocs+\eo-1}}\,\mapsto\,
\delta_{0,j_{\ocs+\eo-1}}\cdot\hat\gb^{l_1,l_2}_{j_0\dots j_{\ocs+\eo-2}}  
\ee
determines a surjective homomorphism $\fds^{\ocs}(\CE,a)\to\fds^{\ocs-1}(\CE,a)$.
Here $\delta_{0,i_{\ocs}}$ and $\delta_{0,j_{\ocs+\eo-1}}$ are the Kronecker deltas.

According to Theorem~\ref{thzcrfd}, the algebra $\fds^{\ocs}(\CE,a)$ 
is responsible for ZCRs of order $\le\ocs$, 
and the algebra $\fds^{\ocs-1}(\CE,a)$ is responsible for ZCRs of order $\le\ocs-1$.    
The constructed homomorphism $\fds^{\ocs}(\CE,a)\to\fds^{\ocs-1}(\CE,a)$ 
reflects the fact that any ZCR of order $\le\ocs-1$ is at the same time of order~$\le\ocs$.
Thus we obtain the following sequence of surjective homomorphisms of Lie algebras
\beq
\lb{fdnn-1}
\dots\to\fds^{\ocs}(\CE,a)\to\fds^{\ocs-1}(\CE,a)\to\dots\to\fds^1(\CE,a)\to\fds^0(\CE,a).
\ee

Other approaches to the study of 
the action of gauge transformations on ZCRs can be found 
in~\cite{marvan93,marvan97,marvan2010,sakov95,sakov2004,sebest2008} 
and references therein.
For a given ZCR with values in a matrix Lie algebra $\mg$, 
the papers~\cite{marvan93,marvan97,sakov95} define 
certain $\mg$-valued functions, which transform by conjugation 
when the ZCR transforms by gauge. 
Applications of these functions to construction and classification of 
some types of ZCRs are described  
in~\cite{marvan93,marvan97,marvan2010,sakov95,sakov2004,sebest2008}.

To our knowledge, 
the theory of~\cite{marvan93,marvan97,marvan2010,sakov95,sakov2004,sebest2008} 
does not produce any infinite-dimensional Lie algebras responsible for ZCRs. 
So this theory does not contain the algebras $\fds^\oc(\CE,a)$.

\section{Generators of the algebras $\fds^{\ocs}(\CE,a)$}
\lb{secrgen}

We continue to study the Lie algebras $\fds^\ocs(\CE,a)$, $\ocs\in\zp$, 
defined in Section~\ref{csev}.
Here $\CE$ is the infinite prolongation of equation~\er{eveq_intr}, 
and $a\in\CE$ is given by~\er{pointevs}.
According to Remark~\ref{rem_fdpgen}, 
the algebra $\fds^\ocs(\CE,a)$ is given by the generators 
\beq
\lb{gagblij}
\ga^{l_1,l_2}_{i_0\dots i_\ocs},\qquad 
\gb^{l_1,l_2}_{j_0\dots j_{\ocs+\eo-1}},\qquad
l_1,l_2,i_0,\dots,i_\ocs,j_0,\dots,j_{\ocs+\eo-1}\in\zp,
\ee
and the relations arising from~\er{xgbtga},~\er{gagb00}.
Using~\er{evdxdt}, we can rewrite equation~\er{xgbtga} as
\beq
\lb{zcrdet}
\frac{\pd}{\pd x}(\gb)+\sum_{k=0}^{\ocs+\eo-1} u_{k+1}\frac{\pd}{\pd u_k}(\gb)
-\frac{\pd}{\pd t}(\ga)-\sum_{k=0}^\ocs
D_x^k\big(F(x,t,u_0,\dots,u_{\eo})\big)\frac{\pd}{\pd u_k}(\ga)+[\ga,\gb]=0,
\ee
where $F(x,t,u_0,\dots,u_{\eo})$ is the right-hand side of equation~\er{eveq_intr}.
We regard $F=F(x,t,u_0,\dots,u_{\eo})$ as a power series, 
using the Taylor series of the function~$F$ at the point~\er{pointevs}.

According to Remark~\ref{rem_fdpgen}, 
the algebra $\fds^\ocs(\CE,a)$ is generated by the elements~\er{gagblij}.
Theorem~\ref{lemgenfdq} says that 
the elements~\er{gal1alprop} generate the algebra $\fds^\ocs(\ce,a)$ as well.
This fact is very useful in computations of $\fds^\ocs(\ce,a)$ 
for concrete equations, because the set of the elements~\er{gal1alprop} 
is much smaller than that of~\er{gagblij}.
Theorem~\ref{lemgenfdq} is used in Section~\ref{sec_we_alg} of this paper
and in the proof of Theorem~\ref{kntkdvtypeth} given in~\cite{sint18}.

\begin{theorem}
\lb{lemgenfdq}
The elements 
\beq
\lb{gal1alprop}
\ga^{l_1,0}_{i_0\dots i_\ocs},\qquad\qquad l_1,i_0,\dots,i_\ocs\in\zp,
\ee
generate the algebra $\fds^\ocs(\ce,a)$.
\end{theorem}
\begin{proof}
For each $l\in\zp$, 
denote by $\agn_l\subset \fds^{\ocs}(\CE,a)$ the subalgebra generated by 
all the elements $\ga^{l_1,l_2}_{i_0\dots i_\ocs}$ with $l_2\le l$.
To prove Theorem~\ref{lemgenfdq}, we need several lemmas. 

\begin{lemma}
\label{gbllalgl}
Let $l_1,l_2,j_0,\dots,j_{\ocs+\eo-1}\in\zp$ be such that $j_0+\dots+j_{\ocs+\eo-1}>0$. 
Then $\gb^{l_1,l_2}_{j_0\dots j_{\ocs+\eo-1}}\in\agn_{l_2}$.
\end{lemma}
\begin{proof}
For any $j_0,\dots,j_{\ocs+\eo-1}\in\zp$ satisfying $j_0+\dots+j_{\ocs+\eo-1}>0$, 
denote by $\Phi(j_0,\dots,j_{\ocs+\eo-1})$ the maximal integer $r\in\{0,1,\dots,\ocs+\eo-1\}$ 
such that $j_r\neq 0$. 
Set also $\Phi(0,\dots,0)=-1$.

Differentiating~\eqref{zcrdet} with respect to $u_{\ocs+\eo}$, we obtain 
$
\frac{\pd}{\pd u_{\ocs+\eo-1}}(\gb)=
\frac{\pd F}{\pd u_{\eo}}\cdot\frac{\pd}{\pd u_\ocs}(\ga),
$
which implies $\gb^{l_1,l_2}_{j_0\dots j_{\ocs+\eo-1}}\in\agn_{l_2}$ 
for all $l_1$, $l_2$, $j_0,\dots,j_{\ocs+\eo-1}$ obeying $\Phi(j_0,\dots,j_{\ocs+\eo-1})=\ocs+\eo-1$.

Let $\ic\in\{0,1,\dots,\ocs+\eo-1\}$ be such that  
\beq
\lb{assumgbn}
\gb^{l_1,l_2}_{j'_0\dots j'_{\ocs+\eo-1}}\in\agn_{l_2}\quad
\text{for all $l_1,l_2,j'_0,\dots,j'_{\ocs+\eo-1}\in\zp$ satisfying   $\Phi(j'_0,\dots,j'_{\ocs+\eo-1})>\ic$}.
\ee  
We are going to show that 
$$
\gb^{l_1,l_2}_{\tilde\jmath_0\dots\tilde\jmath_{\ocs+\eo-1}}\in\agn_{l_2}
$$ 
for all $l_1,l_2,\tilde\jmath_0,\dots,\tilde\jmath_{\ocs+\eo-1}\in\zp$ 
satisfying 
$\Phi(\tilde\jmath_0,\dots,\tilde\jmath_{\ocs+\eo-1})=\ic$.  

For any power series $C$ of the form 
$$
C=\sum_{l_1,l_2,d_0,\dots,d_k\ge 0} 
(x-x_a)^{l_1} (t-t_a)^{l_2}(u_0-a_0)^{d_0}\dots(u_k-a_k)^{d_k}\cdot
C^{l_1,l_2}_{d_0\dots d_k},\qquad\quad C^{l_1,l_2}_{d_0\dots d_k}\in\fds^\ocs(\ce,a),
$$
set 
$$
\ds(C)=\left.\Big(\frac{\pd}{\pd u_{\ic+1}}(C)\Big)\,\right|_{u_k=a_k,\ k\ge \ic+1}.
$$
That is, in order to obtain $\ds(C)$, we differentiate $C$ with respect to $u_{\ic+1}$ 
and then substitute $u_k=a_k$ for all $k\ge \ic+1$. 
Property~\er{gagb00} implies
\beq
\lb{dsga0}
\ds\Big(\frac{\pd}{\pd t}(\ga)\Big)=0.
\ee
Combining~\er{zcrdet} with~\er{dsga0}, we get 
\beq
\lb{dsdxgbsum}
\ds\big(D_x(\gb)\big)=
\ds\bigg(\sum_{k=0}^\ocs D_x^k(F)\frac{\pd}{\pd u_k}(\ga)\bigg)
-\ds\big([\ga,\gb]\big).
\ee

Using~\er{gbsumxt}, one obtains
\begin{multline}
\label{dsdxgb}
\ds\big(D_x(\gb)\big)=
\sum_{\substack{l_1,l_2,j_0,\dots,j_{\ocs+\eo-1}\ge 0,\\
\Phi(j_0,\dots,j_{\ocs+\eo-1})=\ic}}
j_{\ic}(x-x_a)^{l_1}(t-t_a)^{l_2}
(u_0-a_0)^{j_0}\dots(u_{\ic}-a_{\ic})^{j_{\ic}-1}
\gb^{l_1,l_2}_{j_0\dots j_{\ocs+\eo-1}}+\\
+\ds\Bigg(\sum_{\substack{l_1,l_2,j_0,\dots,j_{\ocs+\eo-1}\ge 0,\\
\Phi(j_0,\dots,j_{\ocs+\eo-1})>\ic}}
(t-t_a)^{l_2} 
D_x\Big((x-x_a)^{l_1}(u_0-a_0)^{j_0}\dots(u_{\ocs+\eo-1}-a_{\ocs+\eo-1})^{j_{\ocs+\eo-1}}\Big)
\cdot
\gb^{l_1,l_2}_{j_0\dots j_{\ocs+\eo-1}}\Bigg).
\end{multline}
From~\er{gagb00} it follows that $\ds(\ga)=0$, which yields 
\beq
\label{dsgagb}
\ds\big([\ga,\gb]\big)=\Big[\ds(\ga),\,\gb\,\Big|_{u_k=a_k,\ k\ge \ic+1}\Big]
+\Big[\ga\,\Big|_{u_k=a_k,\ k\ge \ic+1},\,\ds(\gb)\Big]=\Big[\ga\,\Big|_{u_k=a_k,\ k\ge \ic+1},\,\ds(\gb)\Big].
\ee
In view of~\er{dsdxgb},~\er{dsgagb},
for any $l_1,l_2,\tilde\jmath_0,\dots,\tilde\jmath_{\ocs+\eo-1}\in\zp$ 
satisfying  
$\Phi(\tilde\jmath_0,\dots,\tilde\jmath_{\ocs+\eo-1})=\ic$ 
the element $\gb^{l_1,l_2}_{\tilde\jmath_0\dots\tilde\jmath_{\ocs+\eo-1}}$ 
appears only once on the left-hand side of~\er{dsdxgbsum} 
and does not appear on the right-hand side of~\er{dsdxgbsum}. 
Combining~\er{dsdxgbsum},~\er{dsdxgb},~\er{dsgagb},
we see that the element $\gb^{l_1,l_2}_{\tilde\jmath_0\dots\tilde\jmath_{\ocs+\eo-1}}$ 
is equal to a linear combination of elements of the form 
\beq
\label{elemgagb}
\ga^{l_1',l_2'}_{i_0\dots i_\ocs},\quad
\gb^{\hat{l}_1,\hat{l}_2}_{\hat{\jmath}_0\dots\hat{\jmath}_{\ocs+\eo-1}},\quad
\Big[\ga^{l_1',l_2'}_{i_0\dots i_\ocs},
\gb^{\hat{l}_1,\hat{l}_2}_{\hat{\jmath}_0\dots\hat{\jmath}_{\ocs+\eo-1}}\Big],\quad
l_2'\le l_2,\quad \hat{l}_2\le l_2,\quad 
\Phi(\hat{\jmath}_0,\dots,\hat{\jmath}_{\ocs+\eo-1})>\ic.
\ee
Obviously, for any $\hat{l}_2\le l_2$ one has $\agn_{{\hat{l}}_2}\subset\agn_{l_2}$.
Taking into account assumption~\er{assumgbn}, 
we obtain that the elements~\er{elemgagb} belong to $\agn_{{l}_2}$. 
Hence $\gb^{l_1,l_2}_{\tilde\jmath_0\dots\tilde\jmath_{\ocs+\eo-1}}\in\agn_{{l}_2}$.

The proof of the lemma is completed by induction.
\end{proof}

\begin{lemma}
\label{gb00mgl}
For all $l_1,l_2\in\zp$, one has $\gb^{l_1,l_2}_{0\dots 0}\in\agn_{l_2}$.
\end{lemma}
\begin{proof}
According to~\er{gagb00}, we have $\gb^{0,l_2}_{0\dots 0}=0$. 
Therefore, it is sufficient 
to prove $\gb^{l_1,l_2}_{0\dots 0}\in\agn_{l_2}$ for $l_1>0$.

Note that property~\er{gagb00} implies 
\beq
\lb{gatga0}
\ga\Big|_{u_k=a_k,\ k\ge 0}=0,\qquad\qquad
\frac{\pd}{\pd t}(\ga)\,\bigg|_{u_k=a_k,\ k\ge 0}=0.
\ee
In view of~\er{gbsumxt}, one has  
\beq
\lb{pdxgb0}
\frac{\pd}{\pd x}(\gb)\,\bigg|_{u_k=a_k,\ k\ge 0}=
\sum_{l_1>0,\ l_2\ge 0}l_1(x-x_a)^{l_1-1} (t-t_a)^{l_2}\cdot\gb^{l_1,l_2}_{0\dots 0}.
\ee
Substituting $u_k=a_k$ for all $k\in\zp$ in~\er{zcrdet} and 
using~\er{gatga0},~\er{pdxgb0}, we get 
\begin{multline}
\lb{suml1gb0}
\sum_{l_1>0,\ l_2\ge 0} 
l_1(x-x_a)^{l_1-1} (t-t_a)^{l_2}\cdot\gb^{l_1,l_2}_{0\dots 0}=\\
=-\bigg(\sum_{k=0}^{\ocs+\eo-1} u_{k+1}\frac{\pd}{\pd u_k}(\gb)\bigg)
\,\bigg|_{u_k=a_k,\ k\ge 0}
+\bigg(\sum_{k=0}^\ocs D_x^k(F)\frac{\pd}{\pd u_k}(\ga)\bigg)\,\bigg|_{u_k=a_k,\ k\ge 0}.
\end{multline}
Combining~\er{gasumxt},~\er{gbsumxt},~\er{suml1gb0}, 
we see that for any $l_1>0$ and $l_2\ge 0$ the element $\gb^{l_1,l_2}_{0\dots 0}$
is equal to a linear combination of elements of the form 
\beq
\label{gagabj1}
\ga^{l'_1,l_2}_{i_0\dots i_\ocs},\qquad
\gb^{l'_1,l_2}_{j_0\dots j_{\ocs+\eo-1}},\qquad
l'_1,i_0,\dots,i_\ocs,j_0,\dots,j_{\ocs+\eo-1}\in\zp,\qquad 
j_0+\dots+j_{\ocs+\eo-1}=1.
\ee
According to Lemma~\ref{gbllalgl} and the definition of $\agn_{l_2}$, 
the elements~\er{gagabj1} belong to $\agn_{l_2}$. 
Thus $\gb^{l_1,l_2}_{0\dots 0}\in\agn_{l_2}$.
\end{proof}

\begin{lemma}
\label{gallmg}
For all $l_1,l,i_0,\dots,i_\ocs\in\zp$, we have 
$\ga^{l_1,l+1}_{i_0\dots i_\ocs}\in\agn_l$.
\end{lemma}
\begin{proof}
Using~\er{gasumxt}, we can rewrite equation~\er{zcrdet} as
\begin{multline*}
\sum_{l_1,l,i_0,\dots,i_\ocs\ge 0}(l+1) 
(x-x_a)^{l_1} (t-t_a)^{l}(u_0-a_0)^{i_0}\dots(u_\ocs-a_\ocs)^{i_\ocs}\cdot\ga^{l_1,l+1}_{i_0\dots i_\ocs}=\\
=\frac{\pd}{\pd x}(\gb)+\sum_{k=0}^{\ocs+\eo-1} u_{k+1}\frac{\pd}{\pd u_k}(\gb)
-\sum_{k=0}^\ocs D_x^k(F)\frac{\pd}{\pd u_k}(\ga)+[\ga,\gb].
\end{multline*}
This implies that $\ga^{l_1,l+1}_{i_0\dots i_\ocs}$ 
is equal to a linear combination of elements of the form 
\beq
\label{elemgall}
\ga^{\hat{l}_1,\hat{l}_2}_{\hat{\imath}_0\dots\hat{\imath}_\ocs},\quad
\gb^{\tilde{l}_1,\tilde{l}_2}_{\tilde{\jmath}_0\dots\tilde{\jmath}_{\ocs+\eo-1}},\quad
\Big[\ga^{\hat{l}_1,\hat{l}_2}_{\hat{\imath}_0\dots\hat{\imath}_\ocs},
\gb^{\tilde{l}_1,\tilde{l}_2}_{\tilde{\jmath}_0\dots\tilde{\jmath}_{\ocs+\eo-1}}\Big],\quad
\hat{l}_2\le l,\quad \tilde{l}_2\le l,\quad 
\hat{\imath}_0,\dots,\hat{\imath}_\ocs,\tilde{\jmath}_0,\dots,
\tilde{\jmath}_{\ocs+\eo-1}\in\zp.
\ee
Using Lemmas~\ref{gbllalgl},~\ref{gb00mgl} and the condition~$\tilde{l}_2\le l$, we get 
$\gb^{\tilde{l}_1,\tilde{l}_2}_{\tilde{\jmath}_0\dots\tilde{\jmath}_{\ocs+\eo-1}}\in\agn_{\tilde{l}_2}\subset\agn_{l}$. 
Therefore, the elements~\er{elemgall} belong to $\agn_{l}$. 
Hence $\ga^{l_1,l+1}_{i_0\dots i_\ocs}\in\agn_{l}$.
\end{proof}

Now we return to the proof of Theorem~\ref{lemgenfdq}.
According to Lemmas~\ref{gbllalgl},~\ref{gb00mgl} and the definition of~$\agn_{l}$, 
we have 
$\ga^{l_1,l_2}_{i_0\dots i_\ocs},\gb^{l_1,l_2}_{j_0\dots j_{\ocs+\eo-1}}\in\agn_{l_2}$ 
for all $l_1,l_2,i_0,\dots i_\ocs,j_0,\dots,j_{\ocs+\eo-1}\in\zp$. 
Lemma~\ref{gallmg} implies that 
\beq
\notag
\agn_{l_2}\subset\agn_{l_2-1}\subset\agn_{l_2-2}\subset\dots\subset\agn_0.
\ee
Therefore, $\fds^{\ocs}(\CE,a)$ is equal to $\agn_0$, which is generated by the elements~\er{gal1alprop}.
\end{proof}

\section{Relations between $\fds^0(\CE,a)$ and the Wahlquist-Estabrook prolongation algebra}
\lb{sec_we_alg}

Consider a scalar evolution equation of the form
\begin{gather}
\label{gevxt}
u_t
=F(u_0,u_1,\dots,u_{\eo}),\qquad\quad u=u(x,t),\qquad\quad u_k=\frac{\pd^k u}{\pd x^k},\qquad\quad u_0=u. 
\end{gather}
Note that the function $F$ in~\er{gevxt} does not depend on $x$, $t$. 

Let $\CE$ be the infinite prolongation of equation~\er{gevxt}. 
Recall that $x$, $t$, $u_k$ are regarded as coordinates on the manifold $\CE$.
A point $a\in\CE$ is determined by the values of $x$, $t$, $u_k$ at $a$.
Let 
\begin{equation}
\label{pointxt}
a=(x=x_a,\,t=t_a,\,u_k=a_k)\,\in\,\CE,\qquad\qquad x_a,\,t_a,\,a_k\in\fik,\qquad k\in\zp,
\end{equation}
be a point of $\CE$.
The constants $x_a$, $t_a$, $a_k$ are the coordinates 
of~$a$ in the coordinate system $x$, $t$, $u_k$.

The \emph{Wahlquist-Estabrook prolongation algebra} 
of equation~\er{gevxt} at the point~\er{pointxt} 
can be defined in terms of generators and relations as follows.
Consider formal power series 
\beq
\label{wgawgbser}
\wga=\sum_{i\ge 0}(u_0-a_0)^{i}\cdot\wga_{i},\qquad\qquad
\wgb=\sum_{j_0,\dots,j_{\eo-1}\ge 0} 
(u_0-a_0)^{j_0}\dots(u_{\eo-1}-a_{\eo-1})^{j_{\eo-1}}\cdot
\wgb_{j_0\dots j_{\eo-1}},
\ee 
where 
\beq
\lb{wgawgbi}
\wga_{i},\qquad \wgb_{j_0\dots j_{\eo-1}},\quad\qquad i,j_0,\dots,j_{\eo-1}\in\zp, 
\ee
are generators of a Lie algebra, which is described below.
The equation 
\beq
\lb{wxgbtga}
D_x(\wgb)-D_t(\wga)+[\wga,\wgb]=0
\ee
is equivalent to some Lie algebraic relations for~\er{wgawgbi}. 
The Wahlquist-Estabrook prolongation algebra (WE algebra for short) 
is given by the generators~\er{wgawgbi} and the relations arising from~\er{wxgbtga}.
A more detailed definition of the WE algebra is presented in~\cite{mll-2012}.
We denote this Lie algebra by $\wea$.

Then \er{wgawgbser}, \er{wxgbtga} is called 
the \emph{formal Wahlquist-Estabrook ZCR with coefficients in~$\wea$}.

The right-hand side $F=F(u_0,u_1,\dots,u_{\eo})$ 
of~\er{gevxt} appears in equation~\er{wxgbtga}, 
because $F$ appears in the formula 
$D_t=\frac{\pd}{\pd t}+\sum_{k\ge 0} D_x^k(F)\frac{\pd}{\pd u_k}$ 
for the total derivative operator $D_t$.
We are going to show that the algebra $\fds^0(\ce,a)$ for equation~\er{gevxt} 
is isomorphic to some subalgebra of~$\wea$. 

According to Remark~\ref{rem_fdpgen}, the algebra $\fds^0(\ce,a)$ 
is generated by $\ga^{l_1,l_2}_{i}$, $\gb^{l_1,l_2}_{j_0\dots j_{\eo-1}}$.
According to~\er{gagb00}, one has $\ga^{l_1,l_2}_{0}=\gb^{0,l_2}_{0\dots 0}=0$ 
for all $l_1,l_2$.

Since equation~\er{gevxt} is invariant 
with respect to the change of variables $x\mapsto x-x_a,\ t\mapsto t-t_a$, 
we can assume $x_a=t_a=0$ in~\er{pointxt}. 
Since $\ga^{l_1,l_2}_{0}=\gb^{0,l_2}_{0\dots 0}=0$ and $x_a=t_a=0$, 
in the case $\ocs=0$ the power series~\er{gasumxt},~\er{gbsumxt},~\er{xgbtga} 
are written as 
\begin{gather}
\lb{gagab0}
\ga=\sum_{l_1,l_2\ge 0,\ i>0} 
x^{l_1} t^{l_2}(u_0-a_0)^{i}\cdot\ga^{l_1,l_2}_{i},\\
\gb=\sum_{l_1,l_2,j_0,\dots,j_{\eo-1}\ge 0} 
x^{l_1} t^{l_2}(u_0-a_0)^{j_0}\dots(u_{\eo-1}-a_{\eo-1})^{j_{\eo-1}}\cdot
\gb^{l_1,l_2}_{j_0\dots j_{\eo-1}},\qquad\quad 
\gb^{0,l_2}_{0\dots 0}=0,\\
\lb{gagab0zcr}
D_x(\gb)-D_t(\ga)+[\ga,\gb]=0,\qquad\qquad 
\ga^{l_1,l_2}_{i},\,\gb^{l_1,l_2}_{j_0\dots j_{\eo-1}}\in\fds^0(\ce,a).
\end{gather}

The next lemma follows from the definition of $\fds^0(\ce,a)$.
\begin{lemma}
\label{lemfd0zcr}
Let $\bl$ be a Lie algebra.
Consider formal power series of the form  
\begin{gather*}
P=\sum_{l_1,l_2\ge 0,\ i>0} 
x^{l_1} t^{l_2}(u_0-a_0)^{i}\cdot P^{l_1,l_2}_{i},\qquad\quad P^{l_1,l_2}_{i}\in\bl,\\
Q=\sum_{l_1,l_2,j_0,\dots,j_{\eo-1}\ge 0} 
x^{l_1} t^{l_2}(u_0-a_0)^{j_0}\dots(u_{\eo-1}-a_{\eo-1})^{j_{\eo-1}}\cdot
Q^{l_1,l_2}_{j_0\dots j_{\eo-1}},\qquad\  
Q^{l_1,l_2}_{j_0\dots j_{\eo-1}}\in\bl,\quad  
Q^{0,l_2}_{0\dots 0}=0.
\end{gather*}
If $D_x(Q)-D_t(P)+[P,Q]=0$, then 
the map $\ga^{l_1,l_2}_{i}\mapsto P^{l_1,l_2}_{i},\,\ 
\gb^{l_1,l_2}_{j_0\dots j_{\eo-1}}\mapsto Q^{l_1,l_2}_{j_0\dots j_{\eo-1}}$ 
determines a homomorphism from $\fds^0(\ce,a)$ to $\bl$. 
\end{lemma}

Let $\bl$ be a Lie algebra. 
A \emph{formal ZCR of Wahlquist-Estabrook type with coefficients in $\bl$} 
is given by formal power series  
\begin{gather}
\lb{pqmg}
M=\sum_{i\ge 0}(u_0-a_0)^{i}\cdot M_{i},\qquad\quad
N=\sum_{j_0,\dots,j_{\eo-1}\ge 0} 
(u_0-a_0)^{j_0}\dots(u_{\eo-1}-a_{\eo-1})^{j_{\eo-1}}\cdot
N_{j_0\dots j_{\eo-1}},\\
\notag
M_{i},\,N_{j_0\dots j_{\eo-1}}\in\bl,
\end{gather}
satisfying 
\beq
\lb{pqzcr}
D_x(N)-D_t(M)+[M,N]=0.
\ee

The next lemma follows from the definition of the WE algebra~$\wea$.
\begin{lemma}
\label{propwezcr}
Any formal ZCR of Wahlquist-Estabrook type~\er{pqmg},~\er{pqzcr}  
with coefficients in~$\bl$ 
determines a homomorphism $\wea\to\bl$ given by 
$\wga_{i}\mapsto M_{i},\ \wgb_{j_0\dots j_{\eo-1}}\mapsto N_{j_0\dots j_{\eo-1}}$. 
\end{lemma}

\begin{remark}
\lb{remunxi}
For any Lie algebra $\bl$, 
there is a (possibly infinite-dimensional) vector space $V$ such that 
$\bl$ is isomorphic to a Lie subalgebra of $\gl(V)$.
Here $\gl(V)$ is the algebra of linear maps $V\to V$.

For example, one can use the following construction. 
Denote by $\un(\bl)$ the universal enveloping algebra of $\bl$.
We have the injective homomorphism of Lie algebras 
\beq
\notag
\xi\cl \bl\hookrightarrow\gl(\un(\bl)),\quad\qquad \xi(v)(w)=vw,
\quad\qquad v\in \bl,\quad\qquad w\in\un(\bl).
\ee
So one can set $V=\un(\bl)$.
\end{remark}

Denote by $\zf$ the vector space of formal power series in variables $z_1,\,z_2$ 
with coefficients in $\fds^0(\ce,a)$. That is, 
an element of $\zf$ is a power series of the form 
$$
\sum_{l_1,l_2\in\zp}z_1^{l_1}z_2^{l_2}C^{l_1l_2},\qquad\qquad C^{l_1l_2}\in\fds^0(\ce,a).
$$
The space $\zf$ has the Lie algebra structure given by 
$$
\bigg[\sum_{l_1,l_2}z_1^{l_1}z_2^{l_2}C^{l_1l_2},\,
\sum_{\tilde{l}_1,\tilde{l}_2}
z_1^{\tilde{l}_1}z_2^{\tilde{l}_2}
\tilde{C}^{\tilde{l}_1\tilde{l}_2}\bigg]=
\sum_{l_1,l_2,\tilde{l}_1,\tilde{l}_2} z_1^{l_1+\tilde{l}_1}z_2^{l_2+\tilde{l}_2}
\Big[C^{l_1l_2},\tilde{C}^{\tilde{l}_1\tilde{l}_2}\Big],\qquad\quad
C^{l_1l_2},\tilde{C}^{\tilde{l}_1\tilde{l}_2}\in\fds^0(\ce,a).
$$
We have also the following homomorphism of Lie algebras 
\beq
\lb{zffd0}
\nu\cl\zf\to\fds^0(\ce,a),\qquad\qquad
\nu\bigg(\sum_{l_1,l_2\in\zp}z_1^{l_1}z_2^{l_2}C^{l_1l_2}\bigg)=C^{00}. 
\ee

For $i=1,2$, let $\pd_{z_i}\cl \zf\to\zf$ be the linear map given by 
$\pd_{z_i}\big(\sum z_1^{l_1}z_2^{l_2}C^{l_1l_2}\big)=
\sum\frac{\pd}{\pd z_i}\big(z_1^{l_1}z_2^{l_2}\big) C^{l_1l_2}$.
 
Let $\zd$ be the linear span of $\pd_{z_1},\,\pd_{z_2}$ in the vector space of linear maps $\zf\to\zf$. 
Since the maps $\pd_{z_1},\,\pd_{z_2}$ commute, 
the space~$\zd$ is a $2$-dimensional abelian Lie algebra 
with respect to the commutator of maps. 

Denote by $\xtf$ the vector space $\zd\oplus\zf$ with the following Lie algebra structure 
$$
[X_1+f_1,\,X_2+f_2]=X_1(f_2)-X_2(f_1)+[f_1,f_2],\qquad X_1,X_2\in \zd,
\qquad f_1,f_2\in\zf.
$$   
An element of $\xtf$ can be written as a sum of the following form 
$$
\big(y_1\pd_{z_1}+y_2\pd_{z_2}\big)+
\sum z_1^{l_1}z_2^{l_2}C^{l_1l_2},\qquad\quad 
y_1,y_2\in\fik,\qquad C^{l_1l_2}\in\fds^0(\ce,a).
$$

\begin{theorem}
\lb{thmhfd0}
Let $\swe\subset\wea$ be the subalgebra generated by the elements
\beq
\lb{elmh}
(\ad\wga_0)^{k}(\wga_i),\qquad\qquad k\in\zp,\qquad i\in\zsp.
\ee
Then the map $(\ad\wga_0)^{k}(\wga_i)\,\mapsto\,k!\cdot\ga^{k,0}_i$, $k\in\zp$,
determines an isomorphism between $\swe$ and $\fds^0(\CE,a)$.

\textup{(}Note that for $k=0$ we have $(\ad\wga_0)^{0}(\wga_i)=\wga_i$, hence 
$\wga_i\in\swe$ for all $i\in\zsp$.\textup{)}
\end{theorem}
\begin{proof}
We have 
$D_x=\frac{\pd}{\pd x}+\sum_{k\ge 0} u_{k+1}\frac{\pd}{\pd u_k}$ and 
$D_t=\frac{\pd}{\pd t}+\sum_{k\ge 0} D_x^k(F)\frac{\pd}{\pd u_k}$, 
where $F=F(u_0,u_1,\dots,u_{\eo})$ is given in~\er{gevxt}.
Equation~\er{gagab0zcr} is equivalent to
\begin{multline}
\lb{vlong}
\sum_{l_1,l_2,j_0,\dots,j_{\eo-1}} 
\frac{\pd}{\pd x}
\big(x^{l_1} t^{l_2}\big)(u_0-a_0)^{j_0}\dots(u_{\eo-1}-a_{\eo-1})^{j_{\eo-1}}
\cdot
\gb^{l_1,l_2}_{j_0\dots j_{\eo-1}}+\\
+\sum_{l_1,l_2,j_0,\dots,j_{\eo-1}}
x^{l_1} t^{l_2}D_x\Big((u_0-a_0)^{j_0}\dots(u_{\eo-1}-a_{\eo-1})^{j_{\eo-1}}\Big)
\cdot\gb^{l_1,l_2}_{j_0\dots j_{\eo-1}}\\
-\sum_{l_1,l_2,i}\frac{\pd}{\pd t}\big(x^{l_1} t^{l_2}\big)
(u_0-a_0)^{i}\cdot\ga^{l_1,l_2}_{i}-
\sum_{l_1,l_2,i}x^{l_1} t^{l_2}D_t\big((u_0-a_0)^{i}\big)\cdot\ga^{l_1,l_2}_{i}
+[\ga,\gb]=0.
\end{multline}

We regard the expressions 
\begin{gather}
\label{zgaser}
\tilde{\wga}=\pd_{z_1}+\sum_{i>0} (u_0-a_0)^{i}\cdot
\bigg(\sum_{l_1,l_2}z_1^{l_1}z_2^{l_2}\ga^{l_1,l_2}_i\bigg),\\
\lb{zgbser}
\tilde{\wgb}=\bigg(\pd_{z_2}+\sum_{l_1,l_2}z_1^{l_1}z_2^{l_2}\gb^{l_1,l_2}_{0\dots 0}\bigg)
+\sum_{\substack{j_0,\dots,j_{\eo-1}\ge 0,\\ j_0+\dots+j_{\eo-1}>0}}
(u_0-a_0)^{j_0}\dots(u_{\eo-1}-a_{\eo-1})^{j_{\eo-1}}
\cdot\bigg(\sum_{l_1,l_2}z_1^{l_1}z_2^{l_2}
\gb^{l_1,l_2}_{j_0\dots j_{\eo-1}}\bigg)
\end{gather}  
as formal power series with coefficients in~$\xtf$. 

Since the function $F$ in~\er{gevxt} does not depend on $x$ and $t$, 
equation~\er{vlong} is equivalent to 
$$
D_x\big(\tilde{\wgb}\big)-D_t\big(\tilde{\wga}\big)+\big[\tilde{\wga},\tilde{\wgb}\big]=0,
$$ 
which implies that the power series~\er{zgaser},~\er{zgbser} 
constitute a formal ZCR of Wahlquist-Estabrook type with coefficients in~$\xtf$. 

Applying Lemma~\ref{propwezcr} to this formal ZCR, we obtain the homomorphism  
\begin{gather}
\lb{weaxtfa}
\vf\cl\wea\to\xtf,\qquad 
\vf(\wga_0)=\pd_{z_1},\qquad\vf(\wga_i)= 
\sum_{l_1,l_2}z_1^{l_1}z_2^{l_2}\ga^{l_1,l_2}_i,\qquad i>0,\\
\notag
\vf(\wgb_{0\dots 0})=
\bigg(\pd_{z_2}+\sum_{l_1,l_2}z_1^{l_1}z_2^{l_2}\gb^{l_1,l_2}_{0\dots 0}\bigg),
\quad
\vf(\wgb_{j_0\dots j_{\eo-1}})=
\bigg(\sum_{l_1,l_2}z_1^{l_1}z_2^{l_2}
\gb^{l_1,l_2}_{j_0\dots j_{\eo-1}}\bigg),\quad j_0+\dots+j_{\eo-1}>0. 
\end{gather}

Clearly, $\zf$ is a Lie subalgebra of $\xtf=\zd\oplus\zf$. 
In view of~\er{weaxtfa}, for any $k\in\zp$ and $i\in\zsp$ one has
\beq
\label{wga0z1}
\vf\Big((\ad\wga_0)^{k}(\wga_i)\Big)= 
\big(\ad\pd_{z_1}\big)^{k}\bigg(\sum_{l_1,l_2}z_1^{l_1}z_2^{l_2}\ga^{l_1,l_2}_i\bigg)
=\big(\pd_{z_1}\big)^{k}\bigg(\sum_{l_1,l_2}z_1^{l_1}z_2^{l_2}\ga^{l_1,l_2}_i\bigg)\in\zf.
\ee 
Since $\swe\subset\wea$ is generated by the elements~\er{elmh}, 
property~\er{wga0z1} implies $\vf(\swe)\subset\zf\subset\xtf$. 
Using the homomorphism~\er{zffd0} and property~\er{wga0z1}, 
we get 
\beq
\lb{nuvf}
\nu\circ\vf\big|_{\swe}\cl\swe\to\fds^0(\CE,a),\qquad 
(\nu\circ\vf)\Big((\ad\wga_0)^{k}(\wga_i)\Big)
=k!\cdot\ga^{k,0}_i,\qquad k\in\zp,\quad i\in\zsp.
\ee

Using Remark~\ref{remunxi}, we can assume that $\wea$ is embedded 
in the algebra $\gl(V)$ for some vector space~$V$. 
Let $\zs$ be the vector space of power series of the form 
\beq
\lb{xtukc}
\sum_{l_1,l_2,i_0,\dots,i_k\ge 0} 
x^{l_1} t^{l_2}(u_0-a_0)^{i_0}\dots(u_k-a_k)^{i_k}\cdot C^{l_1,l_2}_{i_0\dots i_k},\qquad
C^{l_1,l_2}_{i_0\dots i_k}\in\gl(V),\qquad k\in\zp.
\ee
Note that $\zs$ contains the power series~\er{xtukc} for all $k\in\zp$. 
For each $C\in\zs$, the power series $D_x(C),D_t(C)\in\zs$ are defined according 
to Remark~\ref{psdxdtlie}. 

Recall that $\gl(V)$ consists of linear maps $V\to V$. 
Since $\gl(V)$ is an associative algebra with respect to the composition of maps, 
the space $\zs$ is an associative algebra with respect to the standard multiplication 
of formal power series. 

Also, using Remark~\ref{psdxdtlie} and the Lie bracket on~$\gl(V)$, 
we obtain a Lie bracket on the space $\zs$.  

We set $\wgb_0=\wgb_{0\dots 0}$, where $\wgb_{0\dots 0}$ is the free term 
of the power series $\wgb$ from~\er{wgawgbser}. 
Since $\wga_{i},\,\wgb_{j_0\dots j_{\eo-1}}\in\wea\subset\gl(V)$ for all 
$i,j_0,\dots,j_{\eo-1}\in\zp$, 
the power series $\mathrm{e}^{x\wga_0}$, $\mathrm{e}^{t\wgb_{0}}$,  
and~\er{wgawgbser} belong to~$\zs$. Set 
\beq
\lb{pqdefee}
P=-\mathrm{e}^{t\wgb_0}\wga_0\mathrm{e}^{-t\wgb_0}+
\mathrm{e}^{t\wgb_0}\mathrm{e}^{x\wga_0}\wga\mathrm{e}^{-x\wga_0}\mathrm{e}^{-t\wgb_0},\qquad\quad
Q=-\wgb_0+
\mathrm{e}^{t\wgb_0}\mathrm{e}^{x\wga_0}\wgb\mathrm{e}^{-x\wga_0}\mathrm{e}^{-t\wgb_0}.
\ee
Using~\er{pqdefee}, we get
\begin{gather}
\lb{dxqee}
D_x(Q)=\mathrm{e}^{t\wgb_0}\big[\wga_0,\,\mathrm{e}^{x\wga_0}\wgb\mathrm{e}^{-x\wga_0}\big]
\mathrm{e}^{-t\wgb_0}
+\mathrm{e}^{t\wgb_0}\mathrm{e}^{x\wga_0}D_x(\wgb)\mathrm{e}^{-x\wga_0}\mathrm{e}^{-t\wgb_0},\\
\lb{dtpee}
D_t(P)=-\big[\wgb_0,\,\mathrm{e}^{t\wgb_0}\wga_0\mathrm{e}^{-t\wgb_0}\big]+
\big[\wgb_0,\,\mathrm{e}^{t\wgb_0}\mathrm{e}^{x\wga_0}\wga\mathrm{e}^{-x\wga_0}
\mathrm{e}^{-t\wgb_0}\big]+
\mathrm{e}^{t\wgb_0}\mathrm{e}^{x\wga_0}D_t(\wga)\mathrm{e}^{-x\wga_0}\mathrm{e}^{-t\wgb_0}.
\end{gather}
Recall that $D_x(\wgb)-D_t(\wga)+[\wga,\wgb]=0$ according to~\er{wxgbtga}. 
Combining this with~\er{pqdefee}, \er{dxqee}, \er{dtpee}, one obtains
\beq
\lb{qpwgba}
D_x(Q)-D_t(P)+[P,Q]=
\mathrm{e}^{t\wgb_0}\mathrm{e}^{x\wga_0}\big(D_x(\wgb)-D_t(\wga)+[\wga,\wgb]\big)
\mathrm{e}^{-x\wga_0}\mathrm{e}^{-t\wgb_0}=0.
\ee
Formulas~\er{wgawgbser},~\er{pqdefee} yield
\begin{multline} 
\lb{eeaee}
P=-\mathrm{e}^{t\wgb_0}\wga_0\mathrm{e}^{-t\wgb_0}+
\sum_{i\ge 0}(u_0-a_0)^{i}\cdot\mathrm{e}^{t\wgb_0}\mathrm{e}^{x\wga_0}
\wga_i\mathrm{e}^{-x\wga_0}\mathrm{e}^{-t\wgb_0}=
\\
=\sum_{l_1,l_2\ge 0,\ i>0}x^{l_1}t^{l_2}(u_0-a_0)^{i}  \frac{1}{l_1!l_2!}(\ad\wgb_0)^{l_2}\Big((\ad\wga_0)^{l_1}(\wga_i)\Big).
\end{multline}
\begin{multline} 
\lb{eebee}
Q=-\wgb_0+\sum_{j_0,\dots,j_{\eo-1}\ge 0} 
(u_0-a_0)^{j_0}\dots(u_{\eo-1}-a_{\eo-1})^{j_{\eo-1}}\cdot
\mathrm{e}^{t\wgb_0}\mathrm{e}^{x\wga_0}\wgb_{j_0\dots j_{\eo-1}}
\mathrm{e}^{-x\wga_0}\mathrm{e}^{-t\wgb_0}=\\
=-\wgb_0+\sum_{l_1,l_2,j_0,\dots,j_{\eo-1}\ge 0}  
x^{l_1}t^{l_2}(u_0-a_0)^{j_0}\dots(u_{\eo-1}-a_{\eo-1})^{j_{\eo-1}}
\cdot\frac{1}{l_1!l_2!}(\ad\wgb_0)^{l_2}
\Big((\ad\wga_0)^{l_1}\big(\wgb_{j_0\dots j_{\eo-1}}\big)\Big).
\end{multline}

From~\er{qpwgba},~\er{eeaee},~\er{eebee} it follows that 
the power series $P$, $Q$ satisfy all conditions of Lemma~\ref{lemfd0zcr}. 
Applying Lemma~\ref{lemfd0zcr} to $P$, $Q$ given by~\er{eeaee},~\er{eebee}, we obtain 
the homomorphism 
\begin{gather}
\lb{psiga}
\psi\cl\fds^0(\CE,a)\to\wea,\qquad
\psi\big(\ga^{l_1,l_2}_i\big)=\frac{1}{l_1!l_2!}
(\ad\wgb_0)^{l_2}\Big((\ad\wga_0)^{l_1}(\wga_i)\Big),
\qquad l_1,l_2\in\zp,\quad i\in\zsp,\\ 
\notag
\psi\big(\gb^{l_1,l_2}_{j_0\dots j_{\eo-1}}\big)=
\frac{1}{l_1!l_2!}(\ad\wgb_0)^{l_2}
\Big((\ad\wga_0)^{l_1}(\wgb_{j_0\dots j_{\eo-1}})\Big),
\quad l_1,l_2,j_0,\dots,j_{\eo-1}\in\zp,\quad j_0+\dots+j_{\eo-1}>0,\\
\notag
\psi\big(\gb^{l'_1,l'_2}_{0\dots 0}\big)=\frac{1}{l'_1!l'_2!}(\ad\wgb_0)^{l'_2}
\Big((\ad\wga_0)^{l'_1}(\wgb_{0\dots 0})\Big),\qquad l'_1\in\zsp,\qquad l'_2\in\zp.
\end{gather}
From~\er{psiga} we get 
\beq
\lb{psigal1}
\psi\big(\ga^{l_1,0}_i\big)=\frac{1}{l_1!}
(\ad\wga_0)^{l_1}(\wga_i)\,\in\,\swe,\qquad l_1\in\zp,
\qquad i\in\zsp.
\ee
Since, by Theorem~\ref{lemgenfdq}, the elements 
$\ga^{l_1,0}_{i}$, $l_1\in\zp$, $i\in\zsp$,
generate 
the algebra $\fds^0(\CE,a)$, property~\er{psigal1} implies 
$\psi\big(\fds^0(\CE,a)\big)\subset\swe$. 
Then from~\er{nuvf},~\er{psigal1} it follows that the 
homomorphisms $\psi\cl\fds^0(\CE,a)\to\swe$ and 
$\nu\circ\vf\big|_{\swe}\cl\swe\to\fds^0(\CE,a)$ are inverse to each other.
\end{proof}

\section{The algebras $\fds^\ocs(\CE,a)$ for the KdV equation}
\lb{fdockdv}

We need the following result, which is proved in~\cite{sint18}.
\begin{theorem}[\cite{sint18}]
\lb{kntkdvtypeth}
Let $\CE$ be the infinite prolongation of an equation of the form~\er{utukd} 
with $\kd\in\{1,2,3\}$. Let $a\in\CE$. 
For each $\ocs\in\zsp$, 
consider the surjective homomorphism $\vf_\ocs\cl\fds^\ocs(\CE,a)\to\fds^{\ocs-1}(\CE,a)$ 
from~\er{intfdoc1}. 

If $\ocs\ge\kd+\delta_{\kd,3}$ then  
\beq
\notag
[v_1,v_2]=0\qquad\qquad\forall\,v_1\in\ker\vf_\ocs,\qquad
\forall\,v_2\in\fds^\ocs(\CE,a).
\ee
In other words, if $\ocs\ge\kd+\delta_{\kd,3}$ then 
the kernel of $\vf_\ocs$ is contained in the center of the Lie 
algebra $\fds^\ocs(\CE,a)$.


For each $k\in\zsp$, let 
$\psi_{k}\colon\fds^{k+\kd-1+\delta_{\kd,3}}(\CE,a)\to\fds^{\kd-1+\delta_{\kd,3}}(\CE,a)$ 
be the composition of the homomorphisms
\beq
\notag
\fds^{k+\kd-1+\delta_{\kd,3}}(\CE,a)\to\fds^{k+\kd-2+\delta_{\kd,3}}(\CE,a)\to\dots
\to\fds^{\kd+\delta_{\kd,3}}(\CE,a)\to\fds^{\kd-1+\delta_{\kd,3}}(\CE,a)
\ee
from~\er{intfdoc1}. Then 
$$
[h_1,[h_2,\dots,[h_{k-1},[h_k,h_{k+1}]]\dots]]=0\qquad\quad\forall\,h_1,\dots,h_{k+1}\in\ker\psi_k.  
$$
In particular, the kernel of~$\psi_k$ is nilpotent.
\end{theorem}


\begin{lemma}
\lb{lf0kdv}
Let $\CE$ be the infinite prolongation of the KdV equation 
$u_t=u_3+u_0u_1$. Let $a\in\CE$. 

Then $\fds^0(\CE,a)$ is isomorphic to the direct sum of $\msl_2(\fik[\la])$ 
and a $3$-dimensional abelian Lie algebra.
\textup{(}The Lie algebra $\msl_2(\fik[\la])$ has been defined in 
Section~\ref{subsint1}.\textup{)}
\end{lemma}
\begin{proof}
We are going to use Theorem~\ref{thmhfd0}, 
which says that $\fds^0(\CE,a)$ is isomorphic to a certain subalgebra 
of the Wahlquist-Estabrook prolongation algebra $\wea$.

The Wahlquist-Estabrook prolongation algebra for the KdV equation 
was computed in~\cite{kdv,kdv1}.
According to~\cite{kdv1}, this algebra is isomorphic to the direct sum of 
$\msl_2(\fik[\la])$ and a $5$-dimensional nilpotent Lie algebra $H$.
The algebra $H$ has a basis $(r_{-3},r_{-1},r_0,r_1,r_3)$ satisfying
\beq
\lb{rrel}
[r_1,r_{-1}]=[r_{-3},r_3]=-r_0,\qquad
[r_i,r_j]=0\quad\text{if $\,\,i+j\neq 0$}.
\ee

In order to be in agreement with formulas from~\cite{kdv1}, 
we take the KdV equation in the form
\beq
\lb{kdveck}
u_t=-u_3-12u_0u_1,
\ee
which can be transformed to $u_t=u_3+u_0u_1$ 
by scaling of the variables $x$, $t$, $u$.
(Recall that we use the notation~\er{dnot}.)

Let $(h,y,z)$ be a basis of $\msl_2(\fik)$ satisfying
\beq
\lb{hyzsl2}
[h,y]=2y,\qquad [h,z]=-2z,\qquad [y,z]=h.
\ee
In Section~\ref{sec_we_alg} for any evolution PDE of the form~\er{gevxt} 
we have defined the notion of 
formal Wahlquist-Estabrook ZCR with coefficients in~$\wea$, 
which is given by formulas~\er{wgawgbser},~\er{wxgbtga}.
According to~\cite{kdv1}, for the KdV equation~\er{kdveck}, we have 
$\wea\cong H\oplus\msl_2(\fik[\la])$, and $\wga$ from~\er{wgawgbser} 
can be written as
\beq
\lb{wgakdv}
\wga=-2X_1-2u_0X_2-3u_0^2X_3,
\ee
where
\begin{equation}
\label{xxkdv}
X_1=r_1-\frac12y+\frac12z\la,\qquad 
X_2=r_{-1}+z,\qquad 
X_3=r_{-3}.
\end{equation}
Since $r_1,r_{-1},r_{-3}\in H$ and $y,z,z\la\in\msl_2(\fik[\la])$, 
the elements $X_1$, $X_2$, $X_3$ given by~\er{xxkdv} 
belong to $H\oplus\msl_2(\fik[\la])$.

The paper~\cite{kdv1} uses the symbol $T$ in place of $\la$.
For equation~\er{kdveck}, the paper~\cite{kdv1} contains also 
an explicit formula for $\wgb$ from~\er{wgawgbser}, but it is not needed for us.

Formulas~\er{rrel},~\er{hyzsl2},~\er{wgakdv},~\er{xxkdv} imply that, 
in this case, 
the subalgebra $\swe\subset\wea$ defined in Theorem~\ref{thmhfd0} 
is equal to the subalgebra 
$$
\tilde{H}\oplus\msl_2(\fik[\la])\subset H\oplus\msl_2(\fik[\la])\cong\wea,
$$
where $\tilde{H}\subset H$ is spanned by the elements $r_{-1}$, $r_{-3}$, $r_0$.
Formulas~\er{rrel} imply that the $3$-dimensional Lie algebra $\tilde{H}$ 
is abelian.

According to Theorem~\ref{thmhfd0}, one has $\fds^0(\CE,a)\cong\swe$. 
Since $\swe\cong\tilde{H}\oplus\msl_2(\fik[\la])$, we see that 
$\fds^0(\CE,a)$ is isomorphic to the direct sum of $\msl_2(\fik[\la])$ 
and the $3$-dimensional abelian Lie algebra $\tilde{H}$.
\end{proof}

\begin{theorem}
\lb{thkdvm}
Let $\CE$ be the infinite prolongation of the 
KdV equation $u_t=u_3+u_0u_1$. Let $a\in\CE$. 
Then 
\begin{itemize}
 \item the algebra $\fds^0(\CE,a)$ is isomorphic to the direct sum 
of $\msl_2(\fik[\la])$ and a $3$-dimensional abelian Lie algebra,
\item for every $\ocs\in\zsp$, the algebra $\fds^\ocs(\CE,a)$ 
is obtained from $\msl_2(\fik[\la])$ by applying several times the operation of central extension.
\end{itemize}
\end{theorem}
\begin{proof}
The statement about $\fds^0(\CE,a)$ has been proved in Lemma~\ref{lf0kdv}.
In particular, we see that 
$\fds^0(\CE,a)$ is isomorphic to a central extension of $\msl_2(\fik[\la])$.

For every $\ocs\in\zsp$, consider the surjective homomorphism 
$\vf_\ocs\cl\fds^\ocs(\CE,a)\to\fds^{\ocs-1}(\CE,a)$ from~\er{fdnn-1}.
Since the KdV equation is of the form~\er{utukd} for $\kd=1$, 
Theorem~\ref{kntkdvtypeth} implies that 
the kernel of $\vf_\ocs$ is contained in the center of the Lie algebra 
$\fds^\ocs(\CE,a)$. 
Hence for each $\ocs\in\zsp$ the algebra $\fds^\ocs(\CE,a)$ 
is obtained from $\fds^{\ocs-1}(\CE,a)$ by central extension. 
Since $\fds^0(\CE,a)$ is isomorphic to a central extension of $\msl_2(\fik[\la])$, 
we see that $\fds^\ocs(\CE,a)$ is obtained from $\msl_2(\fik[\la])$ 
by applying several times the operation of central extension.
\end{proof}

\section*{Acknowledgements}

S.~Igonin is a research fellow of 
Istituto Nazionale di Alta Matematica (INdAM), Italy. 
S.~Igonin and G.~Manno are members of GNSAGA of INdAM.

The authors acknowledge support by the project 
``FIR-2013 Geometria delle equazioni differenziali''. 
G.~Manno was also partially supported by 
``Starting grant per giovani ricercatori'' of Politecnico di Torino.
The work of S.~Igonin was carried out within the framework of the State Programme of the Ministry of Education and Science of the Russian Federation, 
project 1.12873.2018/12.1.

S.~Igonin would like to thank A.~P.~Fordy, A.~Henriques, I.~S.~Krasilshchik, 
Yu.~I.~Manin, V.~V.~Sokolov, A.~M.~Verbovetsky, and A.~M.~Vinogradov 
for useful discussions.

S.~Igonin is grateful to the Max Planck Institute for Mathematics (Bonn, Germany) 
for its hospitality and excellent working conditions 
during 02.2006--01.2007 and 06.2010--09.2010, when part of this research was done. 

Also, the authors wish to thank the anonymous referee 
for his/her helpful comments and suggestions, 
many of which have resulted in changes to the revised version of this paper.


\begin{thebibliography}{99}

\bibitem{cdg76}
P.~J.~Caudrey, R.~K.~Dodd, and J.~D.~Gibbon.  
A new hierarchy of Korteweg-de Vries equations.
\emph{Proc. Roy. Soc. London Ser. A} \textbf{351} (1976), 407--422. 

\bibitem{dodd}
R.~Dodd and A.~Fordy.
The prolongation structures of quasipolynomial flows.
\emph{Proc. Roy. Soc. London Ser. A} \textbf{385} (1983), 389--429.

\bibitem{kdv} H.~N.~van Eck.
The explicit form of the Lie algebra of Wahlquist and Estabrook. 
A presentation problem.
\emph{Nederl. Akad. Wetensch. Indag. Math.} \textbf{45} (1983), 149--164.

\bibitem{kdv1} H.~N.~van Eck.
A non-Archimedean approach to prolongation theory.
\emph{Lett. Math. Phys.} \textbf{12} (1986), 231--239.

\bibitem{ft}
L.~D.~Faddeev and L.~A.~Takhtajan. 
\emph{Hamiltonian methods in the theory of solitons.} 
Springer-Verlag, 1987. 

\bibitem{finley93}
J.~D.~Finley and J.~K.~McIver. 
Prolongations to higher jets of Estabrook-Wahlquist coverings for PDEs. 
\emph{Acta Appl. Math.} \textbf{32} (1993), 197--225.

\bibitem{fordy-hh}
A.~P.~Fordy. The H\'enon-Heiles system revisited. 
\emph{Phys. D} \textbf{52} (1991), 204--210.


\bibitem{cfa}
S.~Igonin. 
Coverings and fundamental algebras for partial differential equations. 
\emph{J. Geom. Phys.} \textbf{56} (2006), 939--998. 


\bibitem{mll-2012}
S.~Igonin, J.~van de Leur, G.~Manno, and V.~Trushkov.
Infinite-dimensional prolongation Lie algebras and multicomponent 
Landau-Lifshitz systems associated with higher genus curves. 
\emph{J. Geom. Phys.} \textbf{68} (2013), 1--26. 

\bibitem{zcrm17}
S.~Igonin and G.~Manno. 
On Lie algebras responsible for zero-curvature representations 
of multicomponent (1+1)-dimensional evolution PDEs.
Preprint at arXiv:1703.07217

\bibitem{sbt18}
S.~Igonin and G.~Manno.
On Lie algebras responsible for zero-curvature representations and 
B\"acklund transformations of (1+1)-dimensional scalar evolution PDEs.
Preprint at arXiv:1804.04652

\bibitem{sint18} 
S.~Igonin and G.~Manno.
On Lie algebras responsible for integrability of (1+1)-dimensional scalar evolution PDEs.
Preprint, to appear at arxiv.org


\bibitem{kaup80}
D.~J.~Kaup. 
On the inverse scattering problem for cubic eigenvalue problems 
of the class $\psi _{xxx}+6Q\psi _{x}+6R\psi=\lambda \psi $.
\emph{Stud. Appl. Math.} \textbf{62} (1980), 189--216. 

\bibitem{nonl} 
I.~S.~Krasilshchik and A.~M.~Vinogradov. Nonlocal trends
in the geometry of differential equations. \emph{Acta Appl. Math.}
\textbf{15} (1989), 161--209.

\bibitem{krich80}
I.~M.~Krichever and S.~P.~Novikov.
Holomorphic bundles over algebraic curves and nonlinear equations.
\emph{Russian Math. Surveys} \textbf{35} (1980), 53--79.

\bibitem{marvan93}
M.~Marvan. On zero-curvature representations of partial differential equations. 
\emph{Differential geometry and its applications \textup{(}Opava, 1992\textup{)}}, 103--122.
Silesian Univ. Opava, 1993. 
www.emis.de/proceedings/5ICDGA

\bibitem{marvan97}
M.~Marvan.
A direct procedure to compute zero-curvature representations. The case $sl_2$. 
\emph{Secondary Calculus and Cohomological Physics \textup{(}Moscow, 1997\textup{)}}, 9 pp.
www.emis.de/proceedings/SCCP97

\bibitem{marvan2010}
M.~Marvan. On the spectral parameter problem. 
\emph{Acta Appl. Math.} \textbf{109} (2010), 239--255.

\bibitem{mesh-sok13}
A.~G.~Meshkov and V.~V.~Sokolov.
Integrable evolution equations with constant separant.
\emph{Ufa Math. J.} \textbf{4} (2012), 104--153, 
arXiv:1302.6010

\bibitem{mikh91}
A.~V.~Mikhailov, A.~B.~Shabat, and V.~V.~Sokolov.  
The symmetry approach to classification of integrable equations. 
\emph{What is integrability?}, 115--184. Springer, 1991.

\bibitem{novikov99}
D.~P.~Novikov. 
Algebraic-geometric solutions of the Krichever-Novikov equation. 
\emph{Theoret. and Math. Phys.} \textbf{121} (1999), 1567--1573. 

\bibitem{backl}
C.~Rogers and W.~F.~Shadwick.
\emph{B\"acklund transformations and their applications}. 
Academic Press, New York, 1982.



\bibitem{sakov95}
S.~Yu.~Sakovich.  
On zero-curvature representations of evolution equations. 
\emph{J. Phys. A} \textbf{28} (1995), 2861--2869.

\bibitem{sakov2004}
S.~Yu.~Sakovich. 
Cyclic bases of zero-curvature representations: five illustrations to one concept. 
\emph{Acta Appl. Math.} \textbf{83} (2004), 69--83. 

\bibitem{sand-wang2009}
J.~A.~Sanders and J.~P.~Wang. 
Number theory and the symmetry classification of integrable systems. 
\emph{Integrability}, 89--118, \emph{Lecture Notes in Phys.} \textbf{767}. 
Springer, Berlin, 2009.

\bibitem{sk74}
K.~Sawada and T.~Kotera.
A method for finding {$N$}-soliton solutions 
of the {K}.d.{V}. equation and {K}.d.{V}.-like equation.
\emph{Progr. Theoret. Phys.} \textbf{51} (1974), 1355--1367.

\bibitem{sebest2008}
P.~Sebesty\'en. 
On normal forms of irreducible $sl_n$-valued zero-curvature representations. 
\emph{Rep. Math. Phys.} \textbf{62} (2008), 57--68.


\bibitem{svin-sok83}
S.~I.~Svinolupov, V.~V.~Sokolov, and R.~I.~Yamilov. 
On B\"acklund transformations for integrable evolution equations. 
\emph{Soviet Math. Dokl.} \textbf{28} (1983), 165--168.

\bibitem{Prol} H.~D.~Wahlquist and F.~B.~Estabrook.
Prolongation structures
of nonlinear evolution equations. \emph{J. Math. Phys.}
\textbf{16} (1975), 1--7.

\bibitem{zakh-shab}
V.~E.~Zakharov and A.~B.~Shabat. 
Integration of nonlinear equations of mathematical physics by the method of inverse scattering. II. 
\emph{Functional Anal. Appl.} \textbf{13} (1979), 166--174. 

\end{thebibliography}
\end{document}